\documentclass[final]{siamltex}
\usepackage{amsfonts,amssymb,amsmath,graphicx,subfigure}
\usepackage{afterpage}

\usepackage[pdftex,plainpages=false,pdfpagelabels,hypertexnames=false,
       colorlinks=true,pdfstartview=FitV,
       linkcolor=blue,citecolor=blue,
       urlcolor=blue]{hyperref}
\usepackage{thumbpdf}

\usepackage{braket}
\usepackage{mathrsfs}
\usepackage{leftidx}
\newtheorem{thm}{Theorem}[section]
\newtheorem{cor}[thm]{Corollary}
\newtheorem{lem}[thm]{Lemma}

\newtheorem{rem}[thm]{Remark}
\newtheorem{ass}[thm]{Assumption}
\newtheorem{ex}[thm]{Example}

\numberwithin{table}{section}
\numberwithin{equation}{section}
\numberwithin{figure}{section}

\allowdisplaybreaks   

\graphicspath{
{./figs/}
}

\title{On a Class of Nonlocal Wave Equations \\ from Applications}

\author{Horst Reinhard Beyer\footnotemark[1]~,
Burak Aksoylu\footnotemark[2]~, and Fatih Celiker\footnotemark[3]
\\  \today
}
\begin{document}

\maketitle

\renewcommand{\thefootnote}{\fnsymbol{footnote}} 
\footnotetext[1]{Department of Mathematics, TOBB University of 
Economics and Technology, Ankara, 06560, Turkey
\&
Instituto Tecnol\'ogico Superior de Uruapan, 
Carr. Uruapan-Carapan No. 5555, Col. La Basilia, Uruapan, 
Michoac\'an. M\'exico
\& 
Theoretical Astrophysics, IAAT, Eberhard Karls 
University of T\"ubingen, T\"ubingen 72076, Germany, hbeyer@etu.edu.tr.}

\footnotetext [2] {Department of Mathematics, TOBB University of
  Economics and Technology, Ankara, 06560, Turkey \& Department of
  Mathematics and Statistics, University of New Mexico, Albuquerque,
  NM 87131, USA, baksoylu@etu.edu.tr.  This work was supported in part
  by National Science Foundation DMS 1016190 grant, European Commission
  Marie Curie Career Integration Grant 293978, and Scientific and
  Technological Research Council of Turkey (T\"UB\.{I}TAK) TBAG
  112T240 and MAG 112M891 grants. Research visit of Horst R. Beyer was
  supported in part by T\"UB\.{I}TAK 2221 Fellowship for Visiting
  Scientist Program. Sabbatical visit of Fatih Celiker was supported
  in part by T\"UB\.{I}TAK 2221 Fellowship for Scientist on Sabbatical
  Leave Program.}

\footnotetext[3]{Department of Mathematics, Wayne State University,
  656 W. Kirby, Detroit, MI 48202, USA, celiker@math.wayne.edu. This
  work was supported in part by National Science Foundation DMS
  1115280 grant.} 

\date{\today}

\begin{abstract}
  We study equations from the area of peridynamics, which is an
  extension of elasticity.  The governing equations form a system of
  nonlocal wave equations.  Its governing operator is found to be a
  bounded, linear and self-adjoint operator on a Hilbert space.  We
  study the well-posedness and stability of the associated initial
  value problem. We solve the initial value problem by applying the
  functional calculus of the governing operator. In addition, we give
  a series representation of the solution in terms of spherical Bessel
  functions. For the case of scalar valued functions, the governing
  operator turns out as functions of the Laplace operator.  This
  result enables the comparison of peridynamic solutions to those of
  classical elasticity as well as the introduction of local boundary
  conditions into the nonlocal theory. The latter is studied in a
  companion paper.
\end{abstract}

%%% ----------------------------------------------------------------------

%%% ----------------------------------------------------------------------

%47G10  	Integral operators 
%35L05  	Wave equation
%74B99  	None of the above, but in this section (Mechanics of deformable solids)

%74R05  	Brittle damage
%74R10  	Brittle fracture
%74R99  	None of the above, but in this section

\section{Motivation}

Classical elasticity has been successful in characterizing and
measuring the resistance of materials to crack growth. On the other
hand, peridynamics (PD), a nonlocal extension of continuum mechanics
developed by Silling~\cite{silling2000PDfirstPaper}, is capable of
quantitatively predicting the dynamics of propagating cracks,
including bifurcation.  Its effectiveness has been established in
sophisticated applications such as Kalthoff-Winkler experiments of the
fracture of a steel plate with notches
\cite{kalthoffWinkler1988,silling2003}, fracture and failure of
composites, nanofiber networks, and polycrystal fracture
\cite{kilicMadenci2009,
  oterkusMadenci2012,sillingBobaru2005,sillingBobaru2004}. Further
applications are in the context of multiscale modeling, where PD has
been shown to be an upscaling of molecular dynamics
\cite{selesonGunzburgerParks2014,Seleson:2009:UpscalingMDtoPD} and has
been demonstrated as a viable multiscale material model for length
scales ranging from molecular dynamics to classical elasticity
\cite{askariEtAl2008scidac}.  Also see other related engineering
applications
\cite{celikGuvenMadenci2011,kilicAgwaiMadenci2009,kilicMadenci2010,oterkusMadenciEtAl2012,oterkusMadenci2012_initiation},
the review and news articles
\cite{duEtal2012_sirev,duLipton2014_siamNews,lehoucqSilling2010} for a
comprehensive discussion, and the recent book
\cite{madenciOterkus2014_book}.
%%% End of PD application commercial

We study a class of nonlocal wave equations.  The driving application
is PD.  The same operator is also employed in nonlocal diffusion
\cite{rossiEtAl2010_book,duEtal2012_sirev,selesonGunzburgerParks2013_nonlocalDomains}.
%%% Other nonlocal applications
Similar classes of operators are used in numerous applications such as
population models
\cite{carrilloFife2005,mogilnerEdelstein-Keshet1999}, image processing
\cite{Gilboa:2008:NonlocalImage,kindermannOsherJones2005}, particle
systems \cite{bodnarVelazquez2006}, phase transition
\cite{albertiBellettini1998,albertiBellettini1998-2}, and coagulation
\cite{fournierLaurencot2006}. In addition, we witness a major effort
to meet the need for mathematical theory for PD applications and
related nonlocal problems addressing, for instance, conditioning
analysis, domain decomposition and variational theory
\cite{aksoyluMengesha2010,aksoyluParks2011,aksoyluUnlu2014_nonlocal},
volume constraints
\cite{duEtal2012_sirev,duEtal2013_volumeConstraint,duEtal2013_navierEqu},
nonlinearity
\cite{durukErbayErkip2009,durukErbayErkip2010,durukErbayErkip2011_general,lipton2014},
discretization
\cite{aksoySenocak2011,aksoyluUnlu2014_nonlocal,Emmrich:2007:DiscretizePD_,tianDu2013},
numerical methods
\cite{chenGunzburger2011,duJuTianZhou2013_aPosterioriErrorAnal,duTianZhao2013_adaptiveFEM,selesonBeneddinePrudhomme2013_couplingScheme}, and
various other aspects 
\cite{Alali:2009:HeterogeneousPD,duKammLehoucqParks2012,duZhou2009,emmrichLehoucqPuhst2013,Weckner:2007:PDConverge_,gunzburgerLehoucq2010,hindsRadu2012,lehoucqSilling2008_PD_elasticity,Lehoucq:2008:PDStress,mengesha2012,duMengesha2013_signChangingKernel,mikata2012,selesonGunzburgerParks2014,selesonParks2011_influenceFunction,zhouDu2010}.

It is part of the folklore in physics that the point particle model,
which is the root for \emph{locality} in physics, is the cause of
unphysical singular behavior in the description of the underlying
phenomena.  This fact is a strong indication that, in the long run,
the development of nonlocal theories is necessary for description of
natural phenomena.  Operator theory does not discern the locality or
nonlocality of the governing operator.  This is the strength of this
approach.  This article adds valuable tools to the arsenal of methods
to analyze nonlocal problems, thereby, increasing structural
understanding in the field.

The rest of the article is structured as follows.  We start with a
mathematical introduction in Section \ref{sec:mathIntro}.  In Section
\ref{sec:optreatment}, we set the operator theory framework to treat
the nonlocal wave equation.  We prove basic properties of the
solutions such as well-posedness of the initial value problem and
provide a representation of the solutions in terms of bounded
functions of the governing operator.  We study the stability of
solutions and give conservation laws.  In Section
\ref{sec:properties}, in the vector-valued case, we note that the
governing operator becomes an operator matrix.  The generality of
operator theory allows a simple extension of the results established
for the scalar-valued functions to the vector-valued ones.  We prove
the boundedness of the entries of the governing operator matrix.  The
proof is natural due to operator theory again, because it relies on a
well-known criterion for integral operators.  We present a
``diagonalization'' of the matrix entries. This is accomplished by
employing the unitary Fourier transform and connecting the entries to
maximal multiplication operators.  We study the spectral properties of
the entries.  Then, we reach to a notable result.  Namely, we prove
that \emph{the governing operator is a bounded function of the
  classical local operator}.  This has far reaching
consequences. It enables the incorporation of local boundary
conditions into nonlocal theories, which is the subject of our
companion paper \cite{aksoyluBeyerCeliker2014_bounded}.  We introduce
notion of strong resolvent convergence. This allows us to prove the
convergence of solutions of the governing equation to that of the
classical solution.  We give examples of sequences of micromoduli that
are instance of this result.  In Section \ref{sec:representationSolu},
we consider the calculation of the solution of the wave
equation. Since the governing operator is bounded, holomorphic
functions of that operator can be represented in form of power series
in the operator.  Then, we give a representation of holomorphic
functions, present in the solution of the initial value problem,
utilizing the fact that the governing operator is a sum of two
commuting operators.  We discover that the corresponding power series
can be given in terms of a series of Bessel functions.  In Section
\ref{sec:examples}, we apply the representation in terms of Bessel
functions to special Gaussian micromoduli and Gaussian data.  We
depict the resulting solutions of peridynamic wave equation and
compare to the classical solutions.  We conclude in Section
\ref{sec:conclusion}.

\subsection{Mathematical Introduction} \label{sec:mathIntro}

The formal system of linear peridynamic wave equations in $n$-space dimensions
\cite[Eqn. 54]{silling2000PDfirstPaper} $n \in {\mathbb{N}}^{*}$, is given by 
\begin{equation} \label{formalperidynamicwavequation}
\rho \, \frac{\partial^2 u}{\partial t^2}(x,t) = \int_{{\mathbb{R}}^n} C(x^{\prime}-x) \cdot
\left(u(x^{\prime},t) - u(x,t) \right) dx^{\prime} + b(x,t) \, \, , 
\end{equation}
where ``$\cdot$'' indicates matrix multiplication,
or equivalently by the system
\begin{equation} \label{formalperidynamicwavequation2}
\rho \, \frac{\partial^2 u_j}{\partial t^2}(x,t) = \sum_{k=1}^n \int_{{\mathbb{R}}^n} C_{jk}(x^{\prime}-x) \cdot
\left(u_k(x^{\prime},t) - u_k(x,t) \right) dx^{\prime} + 
b_j(x,t) \, \, ,
\end{equation}
where $x \in {\mathbb{R}}^n$, $t \in {\mathbb{R}}$, $C :
{\mathbb{R}}^n \rightarrow M(n \times n,{\mathbb{R}})$ is the
micromodulus tensor, assumed to be even and assuming values inside the
subspace of symmetric matrices, $\rho > 0$ is the mass density, $b :
{\mathbb{R}}^{n} \times {\mathbb{R}} \rightarrow {\mathbb{R}}^n$ is
the prescribed body force density, and $u : {\mathbb{R}}^{n} \times
{\mathbb{R}} \rightarrow {\mathbb{R}}^n$ is the displacement field.

For comparison, e.g., the corresponding wave equation in classical elasticity in $1$-space dimension is given by 
\begin{equation} \label{classicalelasticity}
\rho \, \frac{\partial^2 u}{\partial t^2} = E \, \frac{\partial^2 u}{\partial x^2} +
b \, \, , 
\end{equation}
where $E > 0$ is the so called ``Young's modulus,'' and describing compression waves in a rod.

If $j, k \in \{1,\dots,n\}$ and $C_{jk} \in L^1({\mathbb{R}}^n)$,
we can rewrite 
(\ref{formalperidynamicwavequation2}) as 
\begin{equation} \label{formalperidynamicwavequation3}
\rho \, \frac{\partial^2 u_{j}}{\partial t^2}(x,t) = 
- \sum_{k=1}^n \left\{ \, 
\left[\int_{{\mathbb{R}}^n} C_{jk}(x^{\prime}) dx^{\prime}\right] u_{k}(x,t) - (C_{jk} * u_{k}(\cdot,t))(x)
\right\} + b_{j}(x,t) \, \, , 
\end{equation}
for all $x \in {\mathbb{R}}^n$, $t \in {\mathbb{R}}$ and 
$j \in \{1,\dots,n\}$
where $*$ denotes the convolution product.  The system
(\ref{formalperidynamicwavequation3}) is the starting point for a
functional analytic interpretation, which leads on a well-posed
initial value problem. For this purpose, we use methods from operator
theory; see, e.g., \cite{beyer2007_book,pazy1983_book}.

\section{Operator-Theoretic Treatment of Systems of Wave Equations}
\label{sec:optreatment}

Analogous to the majority of evolution equations from classical and quantum physics, 
(\ref{formalperidynamicwavequation3}) can be treated with methods from 
operator theory, see, e.g., \cite{beyer2007_book,pazy1983_book} for substantiation of this claim 
and \cite{hutsonPymCloud2005_book} for applications of operator theory in engineering. More specifically, this system falls into 
the class of abstract linear wave equations from Theorem~\ref{abstractwaveequation}. For the proof of this theorem 
see, e.g., 
\cite[Thm. 2.2.1 and Cor. 2.2.2]{beyer2007_book}. Special cases 
of this theorem are proved in \cite{joergens1962,mikhlin1970_book} and 
\cite[Vol.~II]{reedSimon_books}.  Statements and proofs make use of the spectral theorems of (densely-defined, linear and)
self-adjoint operators in Hilbert spaces, including the concept of functions of such operators, 
see, e.g., \cite[Vol.~I]{reedSimon_books},  or standard books on Functional Analysis, 
such as \cite{rudin1991_book,yosida1980_book}. These methods are also used throughout the paper.

This section provides the basic properties of the solutions of
abstract wave equations of the form (\ref{waveequation}).  Some of the
subsequent results are scattered in the literature.  Therefore,
wherever necessary, we provide proofs.  In particular,
Theorem~\ref{abstractwaveequation} gives the well-posedness of the
initial value problem for a class of abstract wave equations,
conservation of energy and a representation of the solutions in terms
of bounded functions of the governing operator. Corollary~\ref{stab1}
and Theorem~\ref{instability} are results on the stability of the
solutions, i.e., their growth for large times.
Theorem~\ref{conservationlaws} provides conservation laws induced by
symmetries of the governing operator.
Theorem~\ref{solutionoftheinhomogenousequation} provides special
solutions of the associated class of inhomogeneous wave equations.
Together with Theorem~\ref{abstractwaveequation}, these solutions
provide the well-posedness of the initial value problem of the latter
equations as well as a representation of the solutions in terms of
bounded functions of the governing operator.

\begin{thm} {\bf (Wave Equations)} \label{abstractwaveequation}
Let $\left( X,\braket{|}\right)$ be some non-trivial 
complex Hilbert space. Furthermore, let $A : D(A) \rightarrow X$ be
some densely-defined, linear, semibounded self-adjoint operator in $X$
with spectrum $\sigma(A)$.
Finally, let $\xi,\eta \in D(A)$. 
\begin{itemize}
\item[(i)]
Then 
there is a unique twice continuously differentiable map 
$u : {\mathbb{R}} \rightarrow X$ assuming values in $D(A)$ and 
satisfying
\begin{equation} \label{waveequation}
u^{\, \prime \prime}(t) = - A \, u(t) 
\end{equation}
for all $t \in {\mathbb{R}}$ 
as well as
\begin{equation*}
u(0) = \xi \, \, , \, \, u^{\, \prime}(0) = \eta \, \, .
\end{equation*}
\item[(ii)]
For this $u$, the corresponding energy function 
$E_u : \mathbb{R} \rightarrow {\mathbb{R}}$,
defined by 
\begin{equation*}
E_u(t) := \frac{1}{2} \, \big( \, \braket{u^{\, \prime}(t)| u^{\, \prime}(t)}
+ \braket{u(t)| A u(t)}  \, \big)  
\end{equation*}
for all $t \in {\mathbb{R}}$,
is constant. 
\item[(iii)]
Moreover, this $u$ is given by 
\begin{equation} \label{representationofthesolution}
u(t) = \left[\overline{\cos \left(t \sqrt{\phantom{ij}} \right)}\,
\bigg|_{\sigma(A)}\right]\!(A) \xi + 
\left[\, \overline{\frac{\sin \left(t \sqrt{\phantom{ij}} \right)}{\sqrt{\phantom{ij}}}} \, \bigg|_{\sigma(A)}\right]\!(A) \eta 
\end{equation} 
for all $t \in {\mathbb{R}}$,
where 
\begin{equation*} 
\overline{\cos(t \sqrt{\phantom{ij}} \,)} \, \, \, \, \textrm{and} \, \, \, \, 
\overline{\frac{\sin(t \sqrt{\phantom{ij}} \,)}{ \sqrt{\phantom{ij}}}}
\end{equation*}
denote the unique extensions of 
$
\cos(t \sqrt{\phantom{ij}} \,) \, \, \, \, \textrm{and}  \, \,
\, \, 
\sin(t \sqrt{\phantom{ij}}) / \sqrt{\phantom{ij}},
$
respectively, 
to entire holomorphic functions. 
\end{itemize}
\end{thm}

Moreover, if $A$ is positive, the solutions of (\ref{waveequation})
are stable, i.e., there are no solutions that are growing
exponentially in the norm.

\begin{cor} \label{stab1} {\bf (Stability of Solutions)}
If $A$ is positive, then 
\begin{equation*}
\|u(t)\| \leqslant \|\xi\| + |t| \cdot \|\eta\|
\end{equation*}
for every $t \in {\mathbb{R}}$.
\end{cor}

\begin{proof}
The statement follows from Theorem~\ref{abstractwaveequation}~(iii), since, from an application of the 
spectral theorem of densely-defined, linear and
self-adjoint operators in Hilbert spaces, it follows that 
the operator norms 
of the operators in (\ref{representationofthesolution})
satisfy
\begin{equation*}
\bigg\| \left[\overline{\cos \left(t \sqrt{\phantom{ij}} \right)}\,
\bigg|_{\sigma(A)}\right]\!(A) \bigg\| \leqslant 1 \, \, , \, \, 
\bigg\|\left[\,\overline{\frac{\sin \left(t \sqrt{\phantom{ij}} \right)}{\sqrt{\phantom{ij}}}} \, \bigg|_{\sigma(A)}\right]\!(A) \bigg\| \leqslant t \, \, , 
\end{equation*}
for every $t \in {\mathbb{R}}$.
\end{proof}

On the other hand, if $A$ is strictly negative, there are solutions of (\ref{waveequation}) that are growing exponentially in the norm. 
The corresponding theorem is not readily found in the literature. 
For the convenience of the reader, we give a proof in the Appendix.

\begin{thm} \label{instability}
{\bf (Instability of Solutions)}
If $\left(X,\braket{|}\right)$, $A : D(A) \rightarrow X$, 
$\sigma(A)$
are as in Theorem~\ref{abstractwaveequation} and, in addition, $A$ is
such that 
\begin{equation*}
\sigma(A) \cap (-\infty,0) \neq \emptyset \, \, , 
\end{equation*}
then there is a twice continuously differentiable map 
assuming values in $D(A)$ and 
satisfying
\begin{equation*} 
u^{\, \prime \prime}(t) = - A \, u(t) 
\end{equation*}
for all $t \in {\mathbb{R}}$ with exponentially growing 
norm.
\end{thm}

\begin{proof}
See the Appendix. 
\end{proof}

The following Theorem~\ref{conservationlaws} can be considered a form
of Noether's Theorem for the solutions of (\ref{waveequation}).  For
the convenience of the reader, we provide a proof in the Appendix.

\begin{thm} \label{conservationlaws} ({\bf Conservation Laws Induced by Symmetries})
Let $u, v : {\mathbb{R}} \rightarrow X$ be twice continuously differentiable map 
assuming values in $D(A)$ and 
satisfying
\begin{equation*} 
u^{\, \prime \prime}(t) = - A \, u(t) \, \, , \, \, 
v^{\, \prime \prime}(t) = - A \, v(t)
\end{equation*}
for all $t \in {\mathbb{R}}$. 
Then the following holds. 

\begin{itemize}
\item[(i)] Then 
$j_{u,v} : {\mathbb{R}} \rightarrow {\mathbb{C}}$, defined by 
\begin{equation*}
j_{u,v}(t) := \braket{u(t)|v^{\prime}(t)} - \braket{u^{\prime}(t)|v(t)} 
\end{equation*}
for every $t \in {\mathbb{R}}$, is constant. 
\item[(ii)] If $B \in L(X,X)$ commutes with $A$, i.e., is such that 
$A \circ B \supset B \circ A$, then 
\begin{equation*}
j_{u,B}(t) := \braket{u(t)|B u^{\prime}(t)} - \braket{u^{\prime}(t)|B u(t)} 
\end{equation*}
for every $t \in {\mathbb{R}}$, is constant.
\item[(iii)] If 
$B$ is a densely-defined, linear self-adjoint operator
in $X$ that commutes with $A$, i.e., is such that every member of its associated spectral
family commutes with every member of the spectral family that is associated to $A$, 
and $u(0), u^{\prime}(0) \in D(A) \cap D(B)$, 
then  $\textrm{Ran}(u), \textrm{Ran}(u^{\prime})  \subset D(A) \cap  D(B)$ and 
\begin{equation*}
j_{u,B}(t) := \braket{u(t)|B u^{\prime}(t)} - \braket{u^{\prime}(t)|B u(t)}
\end{equation*}
for every $t \in {\mathbb{R}}$, is constant.
\end{itemize} 
\end{thm}

\begin{proof}
See the Appendix.
\end{proof}

Duhamel's principle leads to a solution of (\ref{waveequation}) for
vanishing data, the proof of the well-posedness and a representation
of the solutions of the initial value problem of the inhomogeneous
equation,
\begin{equation*} 
u^{\, \prime \prime}(t) = - A \, u(t) + b(t) \, \, , 
\end{equation*}
$t \in {\mathbb{R}}$.  For simplicity, the corresponding subsequent
Theorem~\ref{solutionoftheinhomogenousequation} assumes that $A$ is in
addition positive, which is the most relevant case for applications
because otherwise there are exponentially growing solutions,
indicating that the system is unstable; see Theorem \ref{instability}.
The same statement is true if $\sigma(A)$ is only bounded from below.
On the other hand, Theorem~\ref{solutionoftheinhomogenousequation} can
also be obtained by application of the corresponding well-known more
general theorem for strongly continuous semigroups; see, e.g.,
\cite[Thm. 4.6.2]{beyer2007_book}.  We give a direct proof of
Theorem~\ref{solutionoftheinhomogenousequation} in the Appendix, which
does not rely on methods from the theory of strongly continuous
semigroups. For the definition of weak integration; see, e.g.,
\cite[Sec. 3.2]{beyer2007_book}.

\begin{thm} \label{solutionoftheinhomogenousequation}
{\bf (Solutions of Inhomogeneous Wave Equations)}
Let $\left(X,\braket{|}\right)$, $A : D(A) \rightarrow X$, 
$\sigma(A)$
be as in Theorem~\ref{abstractwaveequation} and, in addition, $A$ be
positive.
Finally, let $f : {\mathbb{R}} \rightarrow X$ be a continuous map, assuming values in $D(A^2)$ such that $Af$, $A^2\!f$ are continuous. Then, $v : {\mathbb{R}} \rightarrow X$,
for every $t \in \mathbb{R}$ defined by 
\begin{align*}
v(t) := \int_{I_t} 
\left[\,\overline{\frac{\sin \left((t - \tau) \sqrt{\phantom{ij}} \right)}{\sqrt{\phantom{ij}}}} \, \bigg|_{\sigma(A)}\!\right]\!(A) f(\tau)
\, d\tau \, \, ,
\end{align*}
where $\int$ denotes weak integration in $X$,
\begin{equation*}
I_{t} := 
\begin{cases}
[0,t] & \text{if $t \geqslant 0$} \\
[t,0] & \text{if $t < 0$}
\end{cases} \, \, , 
\end{equation*}
is twice continuously differentiable, assumes values in $D(A)$,
is such that 
\begin{equation*}
v(0) = v^{\prime}(0) = 0 \, \, ,
\end{equation*} 
and 
\begin{equation*} 
v^{\, \prime \prime}(t) + A \, v(t) = f(t), \quad t \in {\mathbb{R}}.
\end{equation*}
\end{thm}

\begin{proof}
See the Appendix.
\end{proof}

\section{The Governing Operator and Properties}
\label{sec:properties}

The standard data space for the classical wave equation
is a  
$L^2$-space with constant weight, on a non-empty open subset of ${\mathbb{R}}^n$, 
$n \in {\mathbb{N}}^{*}$,  
for instance, $L^2_{\mathbb{C}}({\mathbb{R}})$ for a
bar of infinite extension in $1$-space dimension.   
It turns out that the classical data spaces are
suitable also as data spaces for peridynamics, for instance, again $L^2_{\mathbb{C}}({\mathbb{R}})$ for a
bar
of infinite extension in $1$-space dimension, 
composed of a ``linear peridynamic material.'' This simplifies the discussion of the convergence of peridynamic solutions to classical solutions.

In the following, we represent (\ref{formalperidynamicwavequation3}) in form of 
(\ref{waveequation}), where the governing operator $A$ is
an ``operator matrix,''
consisting of sums of multiples of the identity
and convolution operators, as indicated in (\ref{formalperidynamicwavequation3}). 
These matrix entries will turn out to be pairwise commuting.  
The following remark provides some known relevant information on operator matrices
of bounded operators. On the other hand, we avoid explicit matrix notation. 

\begin{rem} {\bf (Operator Matrices)}
If ${\mathbb{K}} \in \{ {\mathbb{R}}, {\mathbb{C}}
\}$, $(X, \braket{\,|\,})$ a non-trivial ${\mathbb{K}}$-Hilbert space, $(A_{jk})_{j,k \in \{1,\dots,n\}}$ a family of elements of $L(X,X)$.
\begin{itemize} 
\item[(i)] Then by
\begin{equation*}
A(\xi_1,\dots,\xi_n) := \left(\,\sum_{k=1}^n A_{1k} \xi_k,\dots,\sum_{k=1}^n A_{nk} \xi_k\right)
\end{equation*}
for every 
$(\xi_1,\dots,\xi_n) \in X^n$, there is defined a bounded linear operator with adjoint $A^{*}$ given by 
\begin{equation*}
A^{*}(\xi_1,\dots,\xi_n) = \left(\,\sum_{k=1}^n 
A_{k1}^{*}
 \xi_k,\dots,\sum_{k=1}^n A_{kn}^{*} \xi_k\right)
\end{equation*}
for every 
$(\xi_1,\dots,\xi_n) \in X^n$.
\item[(ii)] 
If the members of $(A_{jk})_{j,k \in \{1,\dots,n\}}$
are pairwise commuting, then 
$A$ is bijective if and only if $\det(A)$ is bijective, where 
\begin{align*} 
& \det(A) := 
\sum_{\sigma \in S_n} \textrm{sign}(\sigma) \, A_{1\sigma(1)} 
\cdots A_{n \sigma(n)} \, \, ,
\end{align*}
$S_n$
denotes the set of permutations of $\{1,\dots,n\}$,
\begin{equation*}
\textrm{sign}(\sigma) := 
\prod_{i,j = 1, i<j}^{n} \textrm{sign}(\sigma(j) - \sigma(i))
\end{equation*}
for all $\sigma \in S_n$ and $\textrm{sign}$ denotes the signum function.
\end{itemize}
\end{rem}

The basic properties of the entries of the operator matrix are given
in the following lemma.  In fact, these operators turn out to be
bounded linear operators on $L^2_{\mathbb{C}}({\mathbb{R}}^n)$.
Hence, the boundedness and self-adjointness of $A$ follows from those of
$A_C$.  The boundedness of $A$ has been shown in
\cite{duZhou2009,Weckner:2007:PDConverge_,zhouDu2010} for special class of
kernel functions.  We generalize the result to kernel functions
that are in $L^1({\mathbb{R}}^n)$ by utilizing a well-known criterion
for integral operators; see, e.g., Corollary to
\cite[Thm. 6.24]{weidmann1980_book}.

\begin{lem} \label{governingoperator}
{\bf (Matrix Entries)}
Let $n \in {\mathbb{N}}^{*}$, $\rho > 0$ and 
$C \in L^1({\mathbb{R}}^n)$ be even. Then,
\begin{equation} \label{definitionofA}
A_{C} f := \frac{1}{\rho} \left[
\left(\, 
\int_{{\mathbb{R}}^n} C \, dv^n \right) \!. f - 
C * f
\right] \, \, , 
\end{equation}
for every $f \in L^2_{\mathbb{C}}({\mathbb{R}}^n)$, where
$*$ denotes the convolution product, there is defined a
self-adjoint bounded linear operator on $L^2_{\mathbb{C}}({\mathbb{R}}^n)$ 
with operator norm $\|A_{C}\|$
satisfying
\begin{equation} \label{operatornormofA}
\|A_{C}\| \leqslant \frac{1}{\rho}
\left(\, 
\bigg| \int_{\mathbb{R}} C \, dv^n 
\bigg| + \|C\|_{1}
\right) \leqslant \frac{2 \|C\|_1}{\rho} \, \, .
\end{equation}
\end{lem}

\begin{proof}
For this purpose, we define the projections 
$p_1, p_2 : {\mathbb{R}}^{2n} \rightarrow  {\mathbb{R}}^{n}$
by  
\begin{equation*}
p_1(x_1,\dots,x_n,y_1,\dots,y_n) := (x_1,\dots,x_n) \, \, ,
\, \,  p_2(x_1,\dots,x_n,y_1,\dots,y_n) := (y_1,\dots,y_n)
\end{equation*}
for all $(x_1,\dots,x_n,y_1,\dots,y_n) \in {\mathbb{R}}^{2n}$,
and $K := C \circ (p_1 - p_2)$. Taking into account that 
$C$ is in particular measurable, as a consequence of 
the theory of Lebesgue integration, $K$ is measurable. Also,
since $C$ is even, $K$ is symmetric. Furthermore, for every 
$x \in {\mathbb{R}}^n$ and $y \in {\mathbb{R}}^n$
\begin{equation*}
K(x,\cdot) = C(x - \cdot) = C(\cdot - x) \, \, , \, \, 
K(\cdot,y) = C(\cdot - y) \in L^1({\mathbb{R}}^n)
\end{equation*}
and 
\begin{equation*}
\|K(x,\cdot)\|_1 = \|K(\cdot,y)\|_1 = \|C\|_{1} \, \, .
\end{equation*}
Hence according to a well-known criterion for integral operators
on $L^2$-spaces, see, e.g., 
Corollary to \cite[Thm. 6.24]{weidmann1980_book}, to $K$ there 
corresponds a self-adjoint bounded linear integral operator
${\textrm{Int}(K)}$ on 
$L^2_{\mathbb{C}}({\mathbb{R}}^n)$ with operator norm 
$\leqslant \|C\|_{1}$ and for almost all $x$ 
given by 
\begin{equation*}
[{\textrm{Int}(K)} f](x) = \int_{{\mathbb{R}}^n} K(x,\cdot)  \cdot f \, dv^n = \int_{{\mathbb{R}}^n} C(x - \cdot)  \cdot f \, dv^n = (C * f)(x) \, \, .
\end{equation*}
Hence by (\ref{definitionofA}), there is given a self-adjoint 
bound linear operator $A_{C}$ with operator norm $\|A_{C}\|$
satisfying (\ref{operatornormofA}).

\end{proof}

For the study of the spectral properties of the matrix entries,
needed for the application of the results from Section~\ref{sec:optreatment}, we use Fourier transformations. This step parallels 
the common procedure for constant coefficient differential operators on ${\mathbb{R}}^n, n \in {\mathbb{N}}^{*}$. With the help of the unitary Fourier transform $F_2$,
Theorem~\ref{fouriergov} represents the matrix entries as maximal multiplication operators. This process can be viewed as a form of
``diagonalization'' of the entries. Also, since bounded maximal multiplication 
operators commute, the entries commute pairwise. The spectra of  maximal multiplication operators are well understood, leading to 
Corollary~\ref{spectralpropertiesofthecoefficients}.
Also, 
the functional calculus which is associated to maximal multiplication operators is known and allows the construction of the 
functional calculi of the entries. The latter is used in the proof of Theorem~\ref{classicalgoverningoperator} which proves that   
matrix entries corresponding to spherically symmetric micromoduli
are functions of the Laplace operator. 

\begin{ass}
  In the following, for $n \in {\mathbb{N}}^{*}$, $F_2$ denotes the
  unitary Fourier transformation on $L^2_{\mathbb{C}}({\mathbb{R}}^n)$
  which, for every rapidly decreasing test function 
$f \in {\mathscr S}_{{\mathbb{C}}}({\mathbb{R}})$, is defined by
\begin{equation*}
(F_2 f)(k) := \frac{1}{(2 \pi)^{n/2}} \int_{{\mathbb{R}}^n} e^{-i k \cdot 
{\textrm{id}}_{{\mathbb{R}}^n}} f \, dv^n, \quad  k \in {\mathbb{R}}^n. 
\end{equation*}
Also, we denote by $F_1$ the 
map from $L^1_{\mathbb{C}}({\mathbb{R}}^n)$ to 
$C_{\infty}({\mathbb{R}}^n,{\mathbb{C}})$, the space of continuous functions
on ${\mathbb{R}}^n$ vanishing at infinity, which for every 
$f \in L^1_{\mathbb{C}}({\mathbb{R}}^n)$, is defined by  
\begin{equation*}
(F_1 f)(k) := \int_{{\mathbb{R}}^n} e^{-i k \cdot
{\textrm{id}}_{{\mathbb{R}}^n}} f \, dv^n, \quad  k \in {\mathbb{R}}^n. 
\end{equation*}
\end{ass}

\begin{thm} \label{fouriergov}
{\bf (Fourier Transforms of the Entries)}
Let 
\begin{equation*}
T_{\frac{1}{\rho}[(F_1 C)(0) - F_1 C]}
\end{equation*} 
denote the 
maximal multiplication operator by the bounded continuous function
\begin{equation*}
\frac{1}{\rho}[(F_1 C)(0) - F_1 C] 
\end{equation*}
on $L^2_{\mathbb{C}}({\mathbb{R}}^n)$.
Then 
\begin{equation*}
F_2 \circ A_{C} \circ F_2^{-1} = T_{\frac{1}{\rho}[(F_1 C)(0) - F_1 C]} \, \, .
\end{equation*}
\end{thm}

\begin{proof}
The statement is a consequence of the 
fact that 
\begin{equation*}
[F_2 \circ {\textrm{Int}}(K)] f = 
[T_{F_1 C} \circ F_2] f 
\end{equation*}
for every $f \in L^2_{\mathbb{C}}({\mathbb{R}}^n)$, where 
$K$ and ${\textrm{Int}}(K)$ are defined as in Lemma~\ref{governingoperator} and where $T_{F_1 C}$ denotes the 
maximal multiplication operator on $L^2_{\mathbb{C}}({\mathbb{R}}^n)$ by the bounded continuous 
function $F_1 C$. For the 
proof of this fact, we note that for every $L^1_{\mathbb{C}}({\mathbb{R}}^n) \cap L^2_{\mathbb{C}}({\mathbb{R}}^n)$ 
\begin{align*}
[F_2 \circ {\textrm{Int}}(K)] f & = 
F_2 (C * f) = \frac{1}{(2 \pi)^{n/2}} . F_1 (C * f) 
= \frac{1}{(2 \pi)^{n/2}}  . (F_1 C) (F_1 f) \\
& = 
(F_1 C) (F_2 f) = [T_{F_1 C} \circ F_2] f \, \, . 
\end{align*}
Hence, since $L^1_{\mathbb{C}}({\mathbb{R}}^n) \cap L^2_{\mathbb{C}}({\mathbb{R}}^n)$ is dense in $L^2_{\mathbb{C}}({\mathbb{R}}^n)$,
the bounded linear operators $F_2 \circ {\textrm{Int}}(K)$ 
and $T_{F_1 C} \circ F_2$ coincide on a dense subspace of
$L^2_{\mathbb{C}}({\mathbb{R}}^n)$ and therefore coincide on 
the whole of $L^2_{\mathbb{C}}({\mathbb{R}}^n)$.
\end{proof}

We give the spectrum and point spectrum of $A_C$.

\begin{cor} \label{spectralpropertiesofthecoefficients}
{\bf (Spectral Properties of the Entries)}
\begin{align*}
\sigma(A_{C}) & = \overline{\textrm{Ran} \frac{1}{\rho} . [(F_1 C)(0) - F_1 C]}  \, \, , \\ 
\sigma_{p}(A_C) & =  
\left\{
\lambda \in {\mathbb{R}} : 
\left\{
k \in {\mathbb{R}} : \frac{1}{\rho} . [(F_1 C)(0) - (F_1 C)(k)] 
= \lambda
\right\} \, \textrm{is no Lebesgue null set}
\right\} \, \, , 
\end{align*}
where the overline denotes the closure in ${\mathbb{R}}$.
Finally, for every $\lambda \in \sigma(A_{C})$, $A_{C} - \lambda$
is not surjective.
\end{cor}

\begin{proof}
Let $T_{\frac{1}{\rho}[(F_1 C)(0) - F_1 C]}$ denote
maximal multiplication operator by the bounded continuous function 
$\frac{1}{\rho}[(F_1 C)(0) - F_1 C]$ on $L^2_{\mathbb{C}}({\mathbb{R}}^n)$.
Since $F_2$ is an unitary operator 
\begin{equation*}
F_2 \circ A_{C} \circ F_2^{-1} = T_{\frac{1}{\rho}[(F_1 C)(0) - F_1 C]} \, \, , 
\end{equation*}
where the spectra and the point spectra of $A_{C}$ and
$T_{\frac{1}{\rho}[(F_1 C)(0) - F_1 C]}$ coincide, respectively.
Hence it follows from the properties of maximal multiplication
operators that
\begin{align*}
\sigma(A_{C}) & = \bigg\{ 
\lambda \in {\mathbb{R}} : 
\left(\frac{1}{\rho}[(F_1 C)(0) - F_1 C]\right)^{-1} \! \! (U_{c}(\lambda)) \\
& \qquad \qquad \quad {\textrm{\,\, is no Lebesgue null set for every $c>0$}}
\bigg\} \, \, , \\
\sigma_p(A_{C}) & =  \bigg\{ 
\lambda \in {\mathbb{R}} : 
\left(\frac{1}{\rho}[(F_1 C)(0) - F_1 C]\right)^{-1} \! \! (\lambda) \\
& \qquad \qquad \quad {\textrm{\,\, is no Lebesgue null set}}
\bigg\}  \, \, .
\end{align*}
Since ${\mathbb{R}} \setminus \overline{\textrm{Ran} \frac{1}{\rho} . [(F_1 C)(0) - F_1 C]}$ is open, for $\lambda \in {\mathbb{R}} \setminus \overline{\textrm{Ran} \frac{1}{\rho} . [(F_1 C)(0) - F_1 C]}$, there is $\varepsilon > 0$ such that 
\begin{equation*}
\{ k \in {\mathbb{R}} : 
\frac{1}{\rho} . [(F_1 C)(0) - (F_1 C)(k) \in 
(\lambda - \varepsilon, \lambda + \varepsilon)
\} 
\end{equation*}
is empty, and hence $\lambda \notin \sigma(A_{C})$. On the other hand, since $\frac{1}{\rho} [(F_1 C)(0) - F_1 C]$ is continuous,
for $\lambda \in \textrm{Ran} \frac{1}{\rho} . [(F_1 C)(0) - F_1 C]$ and $c > 0$,
\begin{equation*}
\left(\frac{1}{\rho}[(F_1 C)(0) - F_1 C]\right)^{-1} \! \! (U_{c}(\lambda))
\end{equation*} 
is non-empty and open, hence no Lebesgue null set and 
$\lambda \in \sigma(A)$. Since $\sigma(A_{C})$ is closed, it 
follows that 
\begin{equation*}
\sigma(A_{C}) = \overline{\textrm{Ran} \frac{1}{\rho} . [(F_1 C)(0) - F_1 C]}  \, \, . 
\end{equation*}
Finally, for $\lambda \in {\mathbb{R}}$, since 
\begin{equation*}
F_2 \circ (A_{C} -\lambda) \circ F_2^{-1} = T_{\frac{1}{\rho}[(F_1 C)(0) - F_1 C] - \lambda} 
\end{equation*}
it follows that  
$A_{C} - \lambda$ is surjective if and only if $T_{\frac{1}{\rho}[(F_1 C)(0) - F_1 C] - \lambda}$ is surjective. From the properties of maximal multiplication operators, it follows that the latter 
operator is surjective if and only if it is bijective and hence
if and only if $\lambda \in {\mathbb{R}} \setminus \sigma(A_{C})$. 
\newline
\end{proof}

The notable result we obtained is that the governing operator $A_C$ of
the peridynamic wave equation is a bounded function of the classical
governing operator, present in (\ref{classicalelasticity}). This
observation has far reaching consequences.  It enables the comparison
of peridynamic solutions to those of classical elasticity.  In the
past, only the convergence of the peridynamic operator to the
classical operator has been discussed; see
\cite{aksoyluParks2011,aksoyluUnlu2014_nonlocal,lehoucqSilling2008_PD_elasticity,zhouDu2010}.  More
important for applications is the corresponding convergence of
solutions.  The tool that has been developed for this purpose is the
notion of strong resolvent convergence used in
Theorem~\ref{applicationofstrongresolventconvergence}.

The other remarkable implication is the definition of peridynamic-type
operators on bounded domains as functions of the corresponding
classical operator.  Since the classical operator is defined through
\emph{local} boundary conditions, the functions inherit this
knowledge. This observation opens a gateway to incorporate local
boundary conditions to nonlocal theories, which has vital implications
for numerical treatment of nonlocal problems.  This is the subject of our
companion paper \cite{aksoyluBeyerCeliker2014_bounded}.

\begin{thm} {\bf (A Representation of Matrix Entries Corresponding to Spherically Symmetric Micromoduli
as Functions of the Laplace Operator)} \label{classicalgoverningoperator}
Let $n \in {\mathbb{N}}^{*}$, ${\cal L}_{n}$ be the closure of the positive symmetric, essentially 
self-adjoint operator in $L^2_{\mathbb{C}}({\mathbb{R}}^n)$,
given by  
\begin{equation*}
\left( C_0^{\infty}({\mathbb{R}}^n,{\mathbb{C}})
\rightarrow L^2_{\mathbb{C}}({\mathbb{R}}^n) \, \, , 
\, \, f \mapsto - \frac{E}{\rho} \triangle f
\right) \, \, ,
\end{equation*}
where $\rho > 0$ and $E > 0$. Furthermore, if $n > 1$, in addition,  
let $C$ be spherically symmetric, i.e., 
such that 
\begin{equation*}
C \circ R = C \, \, , 
\end{equation*}
for every $R \in SO(n)$, where $SO(n)$ denotes the map of group of special
orthogonal transformations on ${\mathbb{R}}^n$. Then 
\begin{equation*}
A_{C} = \bigg\{ 
\frac{1}{\rho} \left[ 
(F_1 C)(0) - F_1 C
\right] \circ \iota  \,
\bigg\} ({\cal L}_n) \, \, ,
\end{equation*}
where $\iota : [0,\infty) \rightarrow {\mathbb{R}}^n$ is 
defined by 
\begin{equation*}
\iota(s) := \left( \sqrt{\frac{\rho}{E} \, s\,}\,\right) \!. e_1
\, \, ,
\end{equation*}
for every $s \geqslant 0$
and $e_1, \dots, e_n$ denotes the canonical basis
of ${\mathbb{R}}^n$. 
\end{thm}

\begin{proof}
First, we note that 
\begin{equation*}
F_2 \circ {\cal L}_n \circ F_2^{-1} = T_{\frac{E}{\rho} |\,\,|^2} \,\, ,
\end{equation*}
where $T_{\frac{E}{\rho} |\,\,|^2}$ denotes the maximal 
multiplication operator in $L^2_{\mathbb{C}}({\mathbb{R}}^n)$
by the function $\frac{E}{\rho} |\,\,|^2$. In particular, this 
implies that the spectrum of ${\cal L}_n$, $\sigma({\cal L}_n)$, is given by 
$[0,\infty)$ and for every $g \in U^s_{\mathbb{C}}([0,\infty))$
\footnote{$U^s_{\mathbb{C}}([0,\infty))$ denotes the space of bounded complex-valued functions on $[0,\infty)$ that are strongly measurable in the
sense that they are everywhere $[0,\infty)$
limit of a sequence of step functions.}
that 
\begin{equation*}
g({\cal L}_n) = F_2^{-1} \circ T_{g \circ \left(\frac{E}{\rho} |\,\,|^2\right)} \circ F_{2} \, \, , 
\end{equation*}
where $T_{g \circ \left(\frac{E}{\rho} |\,\,|^2\right)}$ denotes the maximal 
multiplication operator on $L^2_{\mathbb{C}}({\mathbb{R}}^n)$
by the function 
\begin{equation*}
g \circ \left(\frac{E}{\rho} \, |\,\,|^2\right) \, \, .
\end{equation*}
Furthermore, we note that $(F_1 C)(0) - F_1 C \in BC({\mathbb{R}}^n,{\mathbb{R}})$,
where $BC({\mathbb{R}}^n,{\mathbb{R}})$ is the space
of real-valued bounded continuous on ${\mathbb{R}}^n$,  and that $(F_1 C)(0) - F_1 C$ is even, since
for every $k \in {\mathbb{R}}^n$ 
\begin{align*}
(F_1 C)(- k) & = \int_{{\mathbb{R}}^n} e^{i k \cdot {\textrm{id}}_{{\mathbb{R}}^n}} C \, dv^n =
\int_{{\mathbb{R}}^n} e^{- i k \cdot {\textrm{id}}_{{\mathbb{R}}^n}} [C \circ  (-{\textrm{id}}_{{\mathbb{R}}^n})] \, dv^n \\
& =
\int_{{\mathbb{R}}^n} e^{- i k \cdot {\textrm{id}}_{{\mathbb{R}}^n}} C \, dv^n =
(F_1 C)(k) \, \, , \\
(F_1 C)(k) & = \frac{1}{2} \left[ \int_{{\mathbb{R}}^n} e^{- i k \cdot {\textrm{id}}_{{\mathbb{R}}^n}} C \, dv^n + 
\int_{{\mathbb{R}}^n} e^{i k \cdot {\textrm{id}}_{{\mathbb{R}}^n}} C \, dv^n 
\right] \\
& = \int_{{\mathbb{R}}^n} \cos( k \cdot {\textrm{id}}_{{\mathbb{R}}^n}) \, C \, dv^n \,\, , \\
(F_1 C)(0) - (F_1 C)(k) & =
\int_{{\mathbb{R}}^n} \left[1 - \cos( k \cdot {\textrm{id}}_{{\mathbb{R}}^n})\right] C \, dv^n =
2 \int_{{\mathbb{R}}^n} \sin^2\left( \frac{k}{2} \cdot {\textrm{id}}_{{\mathbb{R}}^n}\right) C \, dv^n \, \, . 
\end{align*}
Furthermore for $n > 1$, we note that
\begin{align*}
(F_1 C)(R(k)) & = \int_{{\mathbb{R}}^n} e^{- i R(k) \cdot {\textrm{id}}_{{\mathbb{R}}^n}} C \, dv^n =
 \int_{{\mathbb{R}}^n} e^{- i R(k) \cdot R}  \, (C \circ R) \, dv^n \\
& = \int_{{\mathbb{R}}^n} e^{- i k \cdot {\textrm{id}}_{{\mathbb{R}}^n}} \, (C \circ R) \, dv^n
= \int_{{\mathbb{R}}^n} e^{- i k \cdot {\textrm{id}}_{{\mathbb{R}}^n}} C \, dv^n =
(F_1 C)(k)
\end{align*} 
for every $R \in SO(n)$ and $k \in {\mathbb{R}}^n$ and hence that 
\begin{equation*}
(F_1 C)(k) = (F_1 C)(|k| . e_1)
\end{equation*}
for every $k \in {\mathbb{R}}^n$.
In particular, 
\begin{align*}
\frac{1}{\rho} \left[ 
(F_1 C)(0) - F_1 C
\right] \circ \iota \, \in U^s_{\mathbb{R}}([0,\infty)) 
\end{align*}
and 
\begin{align*}
& \bigg\{\frac{1}{\rho} \left[ 
(F_1 C)(0) - F_1 C
\right] \circ \iota \bigg\}({\cal L}_n) = 
F_2^{-1} \circ T_{\big\{\frac{1}{\rho} \left[ 
(F_1 C)(0) - F_1 C
\right] \circ \iota \big\} \circ 
\left(\frac{E}{\rho} |\,\,|^2\right)
} \circ F_{2} \\
& = F_2^{-1} \circ
T_{\frac{1}{\rho} \left[ 
(F_1 C)(0) - F_1 C
\right] \circ (\,|\,\,| . e_1) 
}
\circ F_{2} = 
F_2^{-1} \circ
T_{\frac{1}{\rho} \left[ 
(F_1 C)(0) - F_1 C
\right]  
}
\circ F_{2} =
A_{C} \, \, .
\end{align*}
\end{proof}

Lemma~\ref{convergenceofboundedfunctions} gives conditions 
for the convergence of bounded functions of a self-adjoint operator 
to converge to that operator, which implies strong resolvent convergence 
and also the strong convergence of the same bounded continuous function of each member of the sequence against that bounded continuous function of the self-adjoint operator; see Theorem~\ref{applicationofstrongresolventconvergence}. 

\begin{lem} \label{convergenceofboundedfunctions} {\bf (Convergence of Bounded Functions of a Self-Adjoint Operator to that Operator)}
Let $(X,\braket{\,|\,})$ be a non-trivial complex Hilbert space and
$A : D(A) \rightarrow X$ a densely-defined, linear and self-adjoint 
operator with spectrum $\sigma(A)$. Furthermore, let $f_1,f_2,\dots$
be a sequence in $U_{\mathbb{C}}^s(\sigma(A))$ that is 
everywhere on $\sigma(A)$ pointwise convergent to ${\textrm{id}}_{\mathbb{\sigma(A)}}$, and for which there is $M > 0$  
such that 
\begin{equation} \label{bound}
|f_{\nu}| \leqslant M [(1 + |\, \,|)|_{\sigma(A)}]  
\end{equation}
for all $\nu \in {\mathbb{R}}$. Then 
\begin{equation*}
\lim_{\nu \rightarrow \infty} f_{\nu}(A)\xi = 
A \xi \, \, , \quad \xi \in D(A).
\end{equation*}
\end{lem}

\begin{proof}
Let $\xi \in D(A)$ and $\psi_{\xi}$
the corresponding spectral measure. According to the spectral 
theorem for densely-defined, self-adjoint linear operators
in Hilbert spaces, ${\textrm{id}}_{\mathbb{R}}^2$ is 
$\psi_{\xi}$-summable and  
\begin{align*}
& \|f_{\mu}(A)\xi - f_{\nu}(A) \xi \|^2 = 
\|(f_{\mu} -f_{\nu})(A)\xi\|^2  \\
& =
\braket{(f_{\mu} -f_{\nu})(A)\xi|(f_{\mu} -f_{\nu})(A)\xi}
= \braket{\xi||f_{\mu} -f_{\nu}|^2(A)\xi} \\
& = \int_{\sigma(A)} |f_{\mu} - f_{\nu}|^2 \, d\psi_{\xi}
= \|f_{\mu} - f_{\nu}\|_{2,\psi_{\xi}}^2 =
\|f_{\mu} - {\textrm{id}}_{\sigma(A)} + {\textrm{id}}_{\sigma(A)} - f_{\nu}\|_{2,\psi_{\xi}}^2 \\
& \leqslant \left(\|f_{\mu} - {\textrm{id}}_{\sigma(A)} \|_{2,\psi_{\xi}} + \|{\textrm{id}}_{\sigma(A)} - f_{\nu}\|_{2,\psi_{\xi}}\right)^2 \, \, ,
\end{align*} 
for $\mu,\nu \in {\mathbb{N}}^{*}$. As a consequence of the pointwise convergence of $f_1,f_2,\dots$ on $\sigma(A)$ to ${\textrm{id}}_{\mathbb{\sigma(A)}}$, 
(\ref{bound}) and Lebesgue's dominated convergence 
theorem, it follows that 
\begin{equation*}
\lim_{\mu \rightarrow \infty} \|f_{\mu} - {\textrm{id}}_{\sigma(A)} \|_{2,\psi_{\xi}} = 0 
\end{equation*}
and hence that $f_{1}(A) \xi,f_{2}(A) \xi,\dots$ is a Cauchy sequence
in $X$. Since $(X,\|\,\,\|)$ is inparticular complete, the latter
implies that $f_{1}(A) \xi,f_{2}(A) \xi,\dots$ is convergent in
$(X,\|\,\,\|)$. Furthermore,
\begin{align*}
& \braket{\xi|\lim_{\nu \rightarrow \infty} f_{\nu}(A) \xi}
= \lim_{\nu \rightarrow \infty} \braket{\xi|f_{\nu}(A) \xi}
= \lim_{\nu \rightarrow \infty} 
\int_{\sigma(A)} f_{\nu} \, d\psi_{\xi} \\
& =
\int_{\sigma(A)} {\textrm{id}}_{\sigma(A)} \, d\psi_{\xi}
= \braket{\xi|A \xi} \, \, , 
\end{align*}
where again the pointwise convergence of $f_1,f_2,\dots$ on $\sigma(A)$ to ${\textrm{id}}_{\mathbb{\sigma(A)}}$,
(\ref{bound}), Lebesgue's dominated convergence 
theorem and the spectral 
theorem for densely-defined, self-adjoint linear operators
in Hilbert spaces has been applied. From the polarization 
identity for $\braket{\,|\,}$, it follows that 
\begin{equation*}
\braket{\xi|\lim_{\nu \rightarrow \infty} f_{\nu}(A) \eta}
= \braket{\xi|A \eta}
\end{equation*}
for all $\xi, \eta \in D(A)$. Since $D(A)$ is dense in $X$, 
the latter implies that 
\begin{equation*}
\braket{\xi|\lim_{\nu \rightarrow \infty} f_{\nu}(A) \eta}
= \braket{\xi|A \eta}
\end{equation*}
for all $\xi \in X$, $\eta \in D(A)$ and hence for every $\eta \in D(A)$ that 
\begin{equation*}
\lim_{\nu \rightarrow \infty} f_{\nu}(A) \eta = A \eta \, \, .
\end{equation*}
\end{proof}

Examples~\ref{example1} and \ref{example2} provide sequences of
micromoduli which satisfy the conditions of
Lemma~\ref{convergenceofboundedfunctions}.  Example~\ref{example1}
has also been treated in \cite{mikata2012,sillingZimmermannAbeyaratne2003}
and Example~\ref{example2} has been treated in \cite{mikata2012}.
Example~\ref{example3}
applies Theorem~\ref{applicationofstrongresolventconvergence} to the
sequences of micromoduli from
Examples~\ref{example1} and~\ref{example2}. As a consequence, for fixed
data and $t \in {\mathbb{R}}$, the solutions of the initial value
problem at time $t$ corresponding to the members of each sequence of
micromoduli converge in $L^2_{\mathbb{C}}({\mathbb{R}})$ to the
corresponding classical solution at time $t$.

\begin{ex} \label{example1}
For every $\nu \in {\mathbb{N}}^{*}$, we define $C_{\nu} 
\in L^1({\mathbb{R}})$ by 
\begin{equation} \label{defexample1}
C_{\nu} := 3 E \nu^3 \chi_{_{\left[-\frac{1}{\nu},\frac{1}{\nu}\right]}} \, \, .
\end{equation} 
For $\nu \in {\mathbb{N}}^{*}$
\begin{equation*}
F_1 C_{\nu} =  6 E \nu^3 \, \overline{
\frac{\sin(\nu^{-1} . {\textrm{id}}_{\mathbb{R}})}{{\textrm{id}}_{\mathbb{R}}}} \, \, , 
\end{equation*}
where 
\begin{equation*}
\overline{
\frac{\sin(\nu^{-1} . {\textrm{id}}_{\mathbb{R}})}{{\textrm{id}}_{\mathbb{R}}}}
\end{equation*}
denotes the unique extension of
$
\sin(\nu^{-1} . {\textrm{id}}_{\mathbb{R}})/{\textrm{id}}_{\mathbb{R}}
$
to a continuous function on ${\mathbb{R}}$. Furthermore, 
for $\nu \in {\mathbb{N}}^{*}$, $\lambda \geqslant 0$
\begin{align*}
& \frac{1}{\rho} \left[ 
(F_1 C_{\nu})(0) - F_1 C
\right] \circ \iota(\lambda) = 
\frac{1}{\rho} \left[ 
(F_1 C_{\nu})(0) - F_1 C_{\nu}
\right]\left( \sqrt{\frac{\rho}{E} \, \lambda \,}\,\right) \\
& = \frac{6E \nu^2}{\rho} \left[ 
1 - \overline{\frac{\sin(\nu^{-1} . {\textrm{id}}_{\mathbb{R}})}{\nu^{-1} .{\textrm{id}}_{\mathbb{R}}}} 
\right]\!\!\left( \sqrt{\frac{\rho}{E} \, \lambda \,}\,\right)
\end{align*}
and $k > 0$
\begin{align*}
& 1 - \frac{\sin(k/\nu)}{k/\nu}  = 
\nu \int_{0}^{1/\nu} [1 - \cos(k x)] \, dx 
= \int_{0}^{1} \left[1 - \cos(k u /\nu)\right] \, du \\
& = \int_{0}^{1} \left[ 
\int_{0}^{k / \nu} u \sin(u y) \, dy\right]
\, du = \frac{k}{\nu} \int_{0}^{1} \left[ 
\int_{0}^{1} u \sin(k u v/ \nu) \, dv\right]
\, du \\
& = \frac{k^2}{\nu^2} \int_{[0,1]^2} u^2 v \, \frac{\sin(k u v/ \nu)}{k u v / \nu} \, 
du dv
\end{align*}
and hence that 
\begin{equation*} 
\nu^2 \left[ 
1 - \frac{\sin(k/\nu)}{k/\nu} \right] 
= k^2 \int_{[0,1]^2} u^2 v \, \frac{\sin(k u v/ \nu)}{k u v / \nu} \, 
du dv \, \, .
\end{equation*}
From the latter, we conclude with the help of Lebesgue's dominated convergence theorem that 
\begin{equation*}
\lim_{\nu \rightarrow \infty} \nu^2 \left[ 
1 - \frac{\sin(k/\nu)}{k/\nu} \right] 
 = \frac{k^2}{6}
\end{equation*}
as well as that 
\begin{equation*}
\bigg| 
\nu^2 \left[ 
1 - \frac{\sin(k/\nu)}{k/\nu} \right] 
\bigg| \leqslant k^2 \int_{[0,1]^2} u^2 v \, \bigg|\frac{\sin(k u v/ \nu)}{k u v / \nu}\bigg| \, 
du dv \leqslant k^2 \int_{[0,1]^2} u^2 v \, 
du dv = \frac{k^2}{6} \, \, . 
\end{equation*}
In particular, we conclude for $\lambda \geqslant 0$ that 
\begin{equation*}
\lim_{\nu \rightarrow \infty} \frac{1}{\rho} \left[ 
(F_1 C_{\nu})(0) - F_1 C
\right] \circ \iota(\lambda) = \frac{6E \nu^2}{\rho} \cdot 
\frac{\rho \lambda}{6 E \nu^2} = \lambda
\end{equation*}
as well as that 
\begin{align*}
\bigg|\frac{1}{\rho} \left[ 
(F_1 C_{\nu})(0) - F_1 C
\right] \circ \iota(\lambda) \bigg| \leqslant  
\frac{6E}{\rho} \cdot \frac{1}{6} \, \frac{\rho \lambda}{E} =
\lambda \, \, .
\end{align*}
Finally, we conclude from  Lemma~\ref{convergenceofboundedfunctions} that 
\begin{equation*}
\lim_{\nu \rightarrow \infty} 
\left\{\frac{1}{\rho} \left[ 
(F_1 C_{\nu})(0) - F_1 C
\right] \circ \iota \right\}\!({\cal L}_1) f =
{\cal L}_1 f
\end{equation*}
for every $f \in D({\cal L}_1) = W^2_{\mathbb{C}}({\mathbb{R}})$, 
where ${\cal L}_1$ is the classical governing operator in $1$ dimension, defined in Theorem~\ref{classicalgoverningoperator}.
\end{ex}

\begin{ex} \label{example2}
For every $\nu \in {\mathbb{N}}^{*}$, we define $C_{\nu} 
\in L^1({\mathbb{R}})$ by 
\begin{equation} \label{defexample2}
C_{\nu} := \frac{2 E \nu^3}{\sqrt{2 \pi} } \, e^{- (\nu^2 / 2) . {\mathrm{id}}_{\mathbb{R}}^2}  = 2 E \nu^2 \cdot 
\frac{\nu}{\sqrt{2 \pi} } \, e^{- (\nu^2 / 2) . {\mathrm{id}}_{\mathbb{R}}^2}
\, \, .
\end{equation} 
For $\nu \in {\mathbb{N}}^{*}$, $\lambda \geqslant 0$
\begin{align*}
& F_1 C_{\nu} =  2 E \nu^2  \cdot  e^{- [1/(2 \nu^2)] . {\mathrm{id}}_{\mathbb{R}}^2} \, \, , \\
& \frac{1}{\rho} \left[ 
(F_1 C_{\nu})(0) - F_1 C_{\nu}
\right] \circ \iota(\lambda) = 
\frac{1}{\rho} \left[ 
(F_1 C_{\nu})(0) - F_1 C_{\nu}
\right]\!\left( \sqrt{\frac{\rho}{E} \, \lambda \,}\,\right) \\
& = \frac{2 E \nu^2}{\rho} \left\{ 
1 - e^{- [1/(2 \nu^2)].  {\mathrm{id}}_{\mathbb{R}}^2}\right\}
\!\!\left( \sqrt{\frac{\rho}{E} \, \lambda \,}\,\right) 
\end{align*}
and $k \geqslant 0$ 
\begin{align*}
\nu^2 [1 - e^{- k^2/(2 \nu^2)}] =
\nu^2 \int_{0}^{k^2/(2 \nu^2)} e^{-u} \, du
= \int_{0}^{k^2/2} e^{- v/\nu^2 } \, dv \, \, .
\end{align*}
From the latter, we conclude for $k \geqslant 0$, 
with the help of Lebesgue's dominated convergence theorem, that 
\begin{equation*}
\lim_{\nu \rightarrow \infty} 
\nu^2 [1 - e^{- k^2/(2 \nu^2)}] = \frac{k^2}{2}
\end{equation*}
as well as that 
\begin{equation*}
\bigg| 
\nu^2 [1 - e^{-k^2/ (2 \nu^2)}]\,
\bigg| 
\leqslant \frac{k^2}{2}
\end{equation*}
and hence for $\lambda \geqslant 0$ that 
\begin{align*}
& \lim_{\nu \rightarrow \infty}
\frac{1}{\rho} \left[ 
(F_1 C_{\nu})(0) - F_1 C_{\nu}
\right] \circ \iota(\lambda) = 
\frac{2 E}{\rho} \frac{\rho}{2 E} \, \lambda 
= \lambda \, \, , \\
& \bigg| \frac{1}{\rho} \left[ 
(F_1 C_{\nu})(0) - F_1 C_{\nu}
\right] \circ \iota(\lambda) \bigg| \leqslant 
\frac{2 E}{\rho} \, \frac{\rho}{2 E} \, \lambda 
= \lambda \, \, .
\end{align*}
Finally, we conclude from  Lemma~\ref{convergenceofboundedfunctions} that 
\begin{equation*}
\lim_{\nu \rightarrow \infty} 
\left\{\frac{1}{\rho} \left[ 
(F_1 C_{\nu})(0) - F_1 C
\right] \circ \iota \right\}\!({\cal L}_1) f =
{\cal L}_1 f
\end{equation*}
for every $f \in D({\cal L}_1) = W^2_{\mathbb{C}}({\mathbb{R}})$, 
where ${\cal L}_1$ is the classical governing operator in $1$ dimension, defined in Theorem~\ref{classicalgoverningoperator}.
\end{ex}

\begin{thm} \label{applicationofstrongresolventconvergence}
{\bf (An Application of Strong Resolvent Convergence)}
Let $(X,\braket{\,|\,})$ be a non-trivial complex Hilbert space
and $A : D(A) \rightarrow X$ a densely-defined, linear and self-adjoint 
operator with spectrum $\sigma(A)$. Furthermore, let $f_1,f_2,\dots$
be a sequence of real-valued functions in $U_{\mathbb{C}}^s(\sigma(A))$ that is 
everywhere on $\sigma(A)$ pointwise convergent to ${\textrm{id}}_{\mathbb{\sigma(A)}}$, and for which there is $M > 0$  
such that 
\begin{equation} \label{bound2}
|f_{\nu}| \leqslant M [(1 + |\, \,|)|_{\sigma(A)}]  
\end{equation}
for all $\nu \in {\mathbb{R}}$. Then for every $g \in BC({{\mathbb{R}},{\mathbb{C}}})$
\begin{equation*}
s-\lim_{\nu \rightarrow \infty} [g|_{\sigma(f_{\nu}(A))}](f_{\nu}(A)) = [g|_{\sigma(A)}](A) \, \, , 
\end{equation*}
where for every $\nu \in {\mathbb{N}}^{*}$, 
$\sigma(f_{\nu}(A))$ denotes the spectrum of 
$f_{\nu}(A)$.
\end{thm}

\begin{proof}
The statement is a consequence of Lemma~\ref{convergenceofboundedfunctions} 
and, for example, 
\cite[Vol.~I, Thm. 8.20 and Thm. 8.25]{reedSimon_books}.
\end{proof}
 
\begin{ex} \label{example3}
As a consequence of Examples~\ref{example1} and \ref{example2},
for every $g \in BC({{\mathbb{R}},{\mathbb{C}}})$
\begin{equation*}
s-\lim_{\nu \rightarrow \infty} [g|_{\sigma(A_{C_{\nu}})}](A_{C_{\nu}})
= [g|_{\sigma({\cal L}_1)}]({\cal L}_1) \, \, , 
\end{equation*}
where ${\cal L}_1$ is the classical governing operator in $1$ dimension, defined in Theorem~\ref{classicalgoverningoperator}, and for every $\nu \in {\mathbb{N}}^{*}$, $A_{C_{\nu}}$ is defined by (\ref{definitionofA}), corresponding to the micromodulus  $C_{\nu}$ given by (\ref{defexample1}) and spectrum $\sigma(A_{C_{\nu}})$, or for every $\nu \in {\mathbb{N}}^{*}$, $A_{C_{\nu}}$ is defined by (\ref{definitionofA}), corresponding to the micromodulus $C_{\nu}$ given by (\ref{defexample2}) and spectrum $\sigma(A_{C_{\nu}})$.
\end{ex}

\section{Representation and Properties of the Solutions}
\label{sec:representationSolu}
We consider the calculation of the solutions of the homogeneous wave
equation using (\ref{representationofthesolution}). Since the
governing peridynamic operator is bounded, the functions of that
operator in (\ref{representationofthesolution}) can be represented in
form of power series in the governing operator.  We provide a
representation of a class of holomorphic functions of a bounded,
self-adjoint operator in
Lemma~\ref{holomorphicfunctionalcalculusI}. We apply this
representation to the functions present in the solution of the initial
value problem of the homogeneous wave equation in
Lemma~\ref{approximations}.
Lemma~\ref{holomorphicfunctionalcalculusI} and
Lemma~\ref{approximations} can be viewed as straightforward
applications of the spectral theorems for densely-defined,
self-adjoint linear operators in Hilbert spaces.  On the other hand,
the matrix entries of the governing operator are sums of two commuting
operators, a multiple of the identity operator and a convolution.
Therefore, power series expansions in terms of the convolution
operator turn out to be more useful.  For this purpose, the
application of the new expansions given in
Theorems~\ref{additiontheorem},~\ref{additiontheorem2} and
\ref{besselrepresentationtheorem} proved to be superior; see
Examples~\ref{gaussians} and \ref{gaussians1}.  In particular,
Corollary~\ref{errorestimates} gives an error estimate for the
expansion in Theorem~\ref{additiontheorem2}. This error estimate has
been used to plot the solution in Figures \ref{fig:comparisonContinuous} 
and \ref{fig:comparisonDiscontinuous}.

\begin{lem} {\bf (Holomorphic Functional Calculus)}
\label{holomorphicfunctionalcalculusI}
Let $(X,\braket{\,|\,})$ be a non-trivial complex 
Hilbert space, $A \in L(X,X)$ self-adjoint and $\sigma(A) \subset {\mathbb{R}}$ the (non-empty, compact) spectrum of $A$. Furthermore, 
let $R > \|A\|$ and  
$f : U_{R}(0) \rightarrow {\mathbb{C}}$ be holomorphic. Then,
the sequence 
\begin{equation*}
\left(\frac{f^{(k)}(0)}{k !} . A^{k} \right)_{k \in {\mathbb{N}}}
\end{equation*}
is absolutely summable in $L(X,X)$ and 
\begin{equation*}
(f|_{{\sigma}(A)})(A) = 
\sum_{k=0}^{\infty} \frac{f^{(k)}(0)}{k !} . A^{k} 
\, \, .
\end{equation*}  
\end{lem}

\begin{proof}
First, we note that according to Taylor's theorem, general properties of power series and the compactness of $\sigma(A)$
that 
\begin{equation*}
\left(\frac{f^{(k)}(0)}{k !} . z^{k} \right)_{k \in {\mathbb{N}}}
\end{equation*}
is absolutely summable for every $z \in U_{R}(0)$ as well as, since 
$\sigma(A) \subset B_{\|A\|}(0) \subset U_{R}(0),$
that the sequence of continuous functions 
\begin{equation*}
\left(\,\sum_{k=0}^{n}  \frac{f^{(k)}(0)}{k !} . {({\textrm{id}}_{\mathbb{R}}}|_{\sigma(A)})^n \right)_{n \in {\mathbb{N}}}  
\end{equation*} 
converges uniformly to the continuous 
function 
$f|_{{\sigma}(A)}$. In particular, since 
$\|A\| < R,$ this implies that the sequence 
\begin{equation*}
\left(\frac{f^{(k)}(0)}{k !} . A^{k} \right)_{k \in {\mathbb{N}}}
\end{equation*}
is absolutely summable in $L(X,X)$, and  
it follows from the spectral theorem for bounded self-adjoint 
operators in Hilbert spaces that 
\begin{equation*}
\left(\,\sum_{k=0}^{n}  \frac{f^{(k)}(0)}{k !} . {({\textrm{id}}_{\mathbb{R}}}|_{\sigma(A)})^n\right)(A) =
\sum_{k=0}^{n} \frac{f^{(k)}(0)}{k !} . A^{k} \, \, , 
\end{equation*}
as well as that 
\begin{equation*}
(f|_{{\sigma}(A)})(A) = 
\sum_{k=0}^{\infty} \frac{f^{(k)}(0)}{k !} . A^{k} 
\, \, .
\end{equation*} 
\end{proof}

\begin{lem} \label{approximations} {\bf (Approximations)}
Let $(X,\braket{\,|\,})$ be a non-trivial complex 
Hilbert space, $\sqrt{\phantom{ij}}$
the complex square-root function, with domain ${\mathbb{C}}
\setminus ((- \infty,0] \times \{0\})$. $A \in L(X,X)$ self-adjoint and $\sigma(A) \subset {\mathbb{R}}$ the 
(non-empty, compact) spectrum of $A$.
For every $t \in {\mathbb{R}}$,
the sequences 
\begin{equation*}
\left( (-1)^{k} \,
\frac{t^{2k} }{(2k)!}\, . A^{k} \right)_{k \in {\mathbb{N}}}
\, \, , \, \, \left( (-1)^{k} \,
\frac{t^{2k+1} }{(2k+1)!} \, . A^{k} \right)_{k \in {\mathbb{N}}}
\end{equation*}
are absolutely summable in $L(X,X)$ and
\begin{align*}
& \left[\overline{\cos \left(t \sqrt{\phantom{ij}} \right)}\,
\bigg|_{\sigma(A)}\right]\!(A) = 
\sum_{k=0}^{\infty} (-1)^{k} \,
\frac{t^{2k} }{(2k)!}\, A^{k} \, \, , \\
& \left[\, \overline{\frac{\sin \left(t \sqrt{\phantom{ij}} \right)}{\sqrt{\phantom{ij}}}} \, \bigg|_{\sigma(A)}\right]\!(A)  = 
\sum_{k=0}^{\infty} (-1)^{k} \,
\frac{t^{2k+1} }{(2k+1)!}\, A^{k} \, \, .
\end{align*} 
\end{lem}

\begin{proof}
We note that
for every $t \in {\mathbb{R}}$
\begin{align*}
& \cos(t \sqrt{\phantom{ij}} \,)  : {\mathbb{C}}
\setminus ((-\infty,0] \times \{0\}) \rightarrow {\mathbb{C}} 
\, \, , \, \, \\
& \cosh(t \sqrt{\phantom{ij}} \circ (-{\textrm{id}}_{\mathbb{C}} )\,)  : {\mathbb{C}}
\setminus ([0,\infty) \times \{0\}) \rightarrow {\mathbb{C}}
\end{align*} 
are holomorphic function such that 
\begin{align*}
& \cos(t \sqrt{z} \,) = \sum_{k=0}^{\infty} (-1)^{k} \,
\frac{(t \sqrt{z})^{2k}}{(2k)!} =
\sum_{k=0}^{\infty} (-1)^{k} \,
\frac{t^{2k} }{(2k)!}\, z^{k} \, \, , \\
& \cosh(t \sqrt{-z} \,) = \sum_{k=0}^{\infty} 
\frac{(t \sqrt{- z})^{2k}}{(2k)!} =
\sum_{k=0}^{\infty} (-1)^{k} \,
\frac{t^{2k} }{(2k)!}\, z^{k} 
\end{align*}
for every $z \in {\mathbb{C}}
\setminus ((- \infty,0] \times \{0\})$ and 
$z \in {\mathbb{C}}
\setminus [0,\infty) \times \{0\})$, respectively. As a consequence,
there is a unique extension of $\cos(t \sqrt{\phantom{ij}} \,)$
to an entire holomorphic function 
$\overline{\cos(t \sqrt{\phantom{ij}} \,)}$  such that
\begin{equation*}
\overline{\cos(t \sqrt{\phantom{ij}} \,)}(z) =
\sum_{k=0}^{\infty} (-1)^{k} \,
\frac{t^{2k} }{(2k)!}\, z^{k} 
\end{equation*}
for every $z \in {\mathbb{C}}$. Furthermore, 
\begin{align*}
& \frac{\sin(t \sqrt{\phantom{ij}} \,)}{ \sqrt{\phantom{ij}}}  : {\mathbb{C}}
\setminus ((-\infty,0] \times \{0\}) \rightarrow {\mathbb{C}} 
\, \, , \, \, \\
& \frac{\sinh(t \sqrt{\phantom{ij}}\,)}{\sqrt{\phantom{ij}}} \circ (-{\textrm{id}}_{\mathbb{C}} )  : {\mathbb{C}}
\setminus ([0,\infty) \times \{0\}) \rightarrow {\mathbb{C}}
\end{align*} 
are holomorphic function such that 
\begin{align*}
& 
 \frac{\sin(t \sqrt{z} \,)}{ \sqrt{z}}
 = \frac{1}{ \sqrt{z}} \, 
 \sum_{k=0}^{\infty} (-1)^{k} \,
\frac{(t \sqrt{z})^{2k+1}}{(2k+1)!} =
\sum_{k=0}^{\infty} (-1)^{k} \,
\frac{t^{2k+1}}{(2k+1)!}\, z^{k} \, \, , \\
& \frac{\sinh(t \sqrt{-z} \,)}{\sqrt{- z}} = 
\frac{1}{ \sqrt{-z}} \, 
\sum_{k=0}^{\infty} 
\frac{(t \sqrt{- z})^{2k+1}}{(2k+1)!} =
\sum_{k=0}^{\infty} (-1)^{k} \,
\frac{t^{2k + 1} }{(2k+1)!}\, z^{k} 
\end{align*}
for every $z \in {\mathbb{C}}
\setminus ((- \infty,0] \times \{0\})$ and 
$z \in {\mathbb{C}}
\setminus [0,\infty) \times \{0\})$, respectively. As a consequence,
there is a unique extension of 
$\sin(t \sqrt{\phantom{ij}} \,)/ \sqrt{\phantom{ij}}$
to an entire holomorphic function
\begin{equation*} 
\overline{\frac{\sin(t \sqrt{\phantom{ij}} \,)}{ \sqrt{\phantom{ij}}}}
\end{equation*}  
such that
\begin{equation*}
\overline{\frac{\sin(t \sqrt{\phantom{ij}} \,)}{ \sqrt{\phantom{ij}}}}\,(z) =
\sum_{k=0}^{\infty} (-1)^{k} \,
\frac{t^{2k+1}}{(2k+1)!}\, z^{k} 
\end{equation*}
for every $z \in {\mathbb{C}}$. In particular, it follows
from Lemma~\ref{holomorphicfunctionalcalculusI} that  
the sequences 
\begin{equation*}
\left( (-1)^{k} \,
\frac{t^{2k} }{(2k)!}\, . A^{k} \right)_{k \in {\mathbb{N}}}
\, \, , \, \, \left( (-1)^{k} \,
\frac{t^{2k+1} }{(2k+1)!} \, . A^{k} \right)_{k \in {\mathbb{N}}}
\end{equation*}
are absolutely summable in $L(X,X)$ and that
\begin{align*}
\left[\cos \left(t \sqrt{\phantom{ij}} \right)
\bigg|_{\sigma(A)}\right]\!(A) & = 
\sum_{k=0}^{\infty} (-1)^{k} \,
\frac{t^{2k} }{(2k)!}\, A^{k} \, \, , \\
\left[\, \overline{\frac{\sin \left(t \sqrt{\phantom{ij}} \right)}{\sqrt{\phantom{ij}}}} \, \bigg|_{\sigma(A)}\right]\!(A)  & = 
\sum_{k=0}^{\infty} (-1)^{k} \,
\frac{t^{2k+1} }{(2k+1)!}\, A^{k} \, \, .
\end{align*}

\end{proof}

In preparation, we provide the power series expansion related to bounded
commuting operators.  We expect that the expansion below can be used 
for large time asymptotic of the solutions of the nonlocal wave equation.

\begin{thm} \label{additiontheorem}
Let $(X,\braket{\,|\,})$ be a non-trivial complex 
Hilbert space, $\sqrt{\phantom{ij}}$
the complex square-root function, with domain ${\mathbb{C}}
\setminus ((- \infty,0] \times \{0\})$. $A, B \in L(X,X)$ self-adjoint
such that $[A,B] = 0$ and $\sigma(A), \sigma(A+B) \subset {\mathbb{R}}$ the 
(non-empty, compact) spectra of $A$ and $A+B$, respectively.
Then
\begin{align*}
& \left[\overline{\cos \left(t \sqrt{\phantom{ij}} \right)}\,
\bigg|_{\sigma(A+B)}\right]\!(A + B) \\
& = \sum_{k=0}^{\infty} (-1)^{k} \cdot \frac{t^{2k}}{(2k)!} \, . 
\left\{
\left[\leftidx{_0}{F}{_1}\!\!\left(-; k + \frac{1}{2}; - \, \frac{t^2}{4} \, . {\textrm id}_{\sigma(A)}\right)\right]\!\!(A)
\right\} \! B^k
\, \, , \\
& \left[\, \overline{\frac{\sin \left(t \sqrt{\phantom{ij}} \right)}{\sqrt{\phantom{ij}}}} \, \bigg|_{\sigma(A+B)}\right]\!(A+B)  \\
& = \sum_{k=0}^{\infty} (-1)^{k} \cdot \frac{t^{2k+1}}{(2k+1)!} \, . 
\left\{
\left[\leftidx{_0}{F}{_1}\!\!\left(-; k + \frac{3}{2}; - \, \frac{t^2}{4} \, . {\textrm id}_{\sigma(A)}\right)\right]\!\!(A)
\right\} \! B^k
\, \, , 
\end{align*} 
where $\leftidx{_0}{F}{_1}$ denotes the generalized hypergeometric
function, defined as in \cite{olverEtAl2010_book}.
\end{thm}

\begin{proof}
In a first step, we note for every $t \in {\mathbb{R}}$
that
the family 
\begin{equation*}
\left((-1)^{k+l} \binom{k+l}{l} \,
\frac{t^{2(k+l)} }{[2(k+l)]!}\, A^{k} B^{l}\right)_{(k,l) \in {\mathbb{N}}^2}
\end{equation*}
is absolutely summable in $L(X,X)$, since for 
$(k,l) \in {\mathbb{N}}^2$
\begin{align*}
&
\bigg\| (-1)^{k+l} \binom{k+l}{l} \,
\frac{t^{2(k+l)} }{[2(k+l)]!}\, A^{k} B^{l} \bigg\|
\leqslant 
\frac{t^{2(k+l)} }{[2(k+l)]!} \binom{k+l}{l} \|A\|^k \|B\|^l \\
& = 
\frac{(k+l)!}{l! k! [2(k+l)]!}
\left(\,t^2 \|A\|\,\right)^{\!k} \left(\,t^2 \|B\|\,\right)^{\!l}
\leqslant  \frac{1}{k!} \left(\,t^2\|A\|\,\right)^{\!k} \frac{1}{l!} (\,t^2\|B\|\,)^{l} 
\end{align*} 
and hence for every finite subset $S \subset {\mathbb{N}}^2$
\begin{align*}
& \sum_{(k,l) \in S} \bigg\| (-1)^{k+l} \binom{k+l}{l} \,
\frac{t^{2(k+l)} }{[2(k+l)]!}\, A^{k} B^{l} \bigg\|
\leqslant \exp\left(t^2\|A\| 
\right) \exp\left(t^2\|B\|\right)  \, \, .
\end{align*}
Also, we note 
that
the family 
\begin{equation*}
\left((-1)^{k+l} \binom{k+l}{l} \,
\frac{t^{2(k+l)+1} }{[2(k+l)+1]!}\, A^{k} B^{l} \right)_{(k,l) \in {\mathbb{N}}^2}
\end{equation*}
is absolutely summable in $L(X,X)$, since for 
$(k,l) \in {\mathbb{N}}^2$
\begin{align*}
& \bigg\|(-1)^{k+l} \binom{k+l}{l} \,
\frac{t^{2(k+l)+1} }{[2(k+l)+1]!}\, A^{k} B^{l} \bigg\|_2 
\leqslant 
\frac{|t|^{2(k+l)+1} }{[2(k+l)+1]!} \binom{k+l}{l} \|A\|^k \|B\|^l \\
& = 
|t| \, \frac{(k+l)!}{l! k! [2(k+l)+1]!}
\left(\,t^2 \|A\|\,\right)^{\!k} \left(\,t^2 \|B\|\,\right)^{\!l}
\leqslant  |t| \, \frac{1}{k!} \left(\,t^2\|A\|\,\right)^{\!k} \frac{1}{l!} (\,t^2\|B\|\,)^{l} \, \, , 
\end{align*}
leading to 
\begin{align*}
& \sum_{(k,l) \in S}
\left\|(-1)^{k+l} \binom{k+l}{l} \,
\frac{t^{2(k+l)+1} }{[2(k+l)+1]!}\, A^{k} B^{l}\right\|
\leqslant |t| \exp\left(t^2\|A\| 
\right) \exp\left(t^2\|B\|\right) \, \, , 
\end{align*}
for every finite subset $S \subset {\mathbb{N}}^2$. Hence, we conclude 
the following. 
\begin{align*}
& \sum_{k=0}^{\infty} (-1)^{k} \,
\frac{t^{2k} }{(2k)!}\, (A + B)^{k} =
\sum_{k=0}^{\infty} \sum_{l=0}^{k} (-1)^{k} \,
\frac{t^{2k} }{(2k)!}\, \binom{k}{l} A^{k - l} B^{l} \\
& = \sum_{l=0}^{\infty} \sum_{k=l}^{\infty} (-1)^{k} \,
\frac{t^{2k} }{(2k)!}\, \binom{k}{l} A^{k - l} B^{l} =
\sum_{l=0}^{\infty} \left[ \sum_{k=l}^{\infty} (-1)^{k} \, \binom{k}{l}
\frac{t^{2k} }{(2k)!}\, A^{k - l} \right] \! B^{l}  \\
& = \sum_{l=0}^{\infty} \left[ \sum_{k=0}^{\infty} (-1)^{k+l} \binom{k+l}{l} \,
\frac{t^{2(k+l)} }{[2(k+l)]!}\, A^{k} \right] \! B^{l} \\
& = 
\sum_{l=0}^{\infty} (-1)^{l} t^{2l} 
\left[ \sum_{k=0}^{\infty} (-1)^{k} \binom{k+l}{l} \,
\frac{t^{2k}}{[2(k+l)]!}\, A^{k} \right] \! B^{l} \\
& = \sum_{l=0}^{\infty} (-1)^{l} \, t^{2l} 
\left[ \sum_{k=0}^{\infty} \frac{(k+l)!}{[2(k+l)]! \cdot l! } \cdot 
\frac{1}{k!} \, 
(- t^2 A)^{k} \right] \! B^{l} 
\end{align*}
In the following, we show the auxiliary result that 
for every $k, l \in {\mathbb{N}}$, 
\begin{equation} \label{auxgeneralizedhypergeometric}
\frac{(k+l)!}{[2(k+l)]! \cdot l!}  = 
2^{-k} \cdot \frac{1}{\prod_{m=0}^{k-1}[2(l+m) + 1]} \cdot \frac{1}{(2l)!}
\end{equation}
The proof proceeds by induction on $k$. First, we note that 
\begin{equation*}
\frac{l!}{(2l)! \cdot l!}  = 
2^{-0} \cdot \frac{1}{\prod_{m=0}^{-1}[2(l+m) + 1]} \cdot \frac{1}{(2l)!} \, \, .
\end{equation*}
In the following, we assume that (\ref{auxgeneralizedhypergeometric})
is true for some $k \in {\mathbb{N}}$. Then
\begin{align*}
& \frac{(k+l+1)!}{[2(k+l+1)]! \cdot l!}  = \frac{1}{2} \cdot \frac{1}{2(k+l)+1} \cdot
\frac{(k+l)!}{[2(k+l)]! \cdot l!}  \\
& = \frac{1}{2} \cdot \frac{1}{2(k+l)+1} \cdot 2^{-k} \cdot \frac{1}{\prod_{m=0}^{k-1}[2(l+m) + 1]} \cdot \frac{1}{(2l)!} \\
& = 2^{-(k+1)} \cdot \frac{1}{\prod_{m=0}^{k}[2(l+m) + 1]} \cdot \frac{1}{(2l)!} \, \, , 
\end{align*}
and hence (\ref{auxgeneralizedhypergeometric})
is true for $k+1$. The equality (\ref{auxgeneralizedhypergeometric})
implies for every $k, l \in {\mathbb{N}}$ that 
\begin{align*}
& \frac{(k+l)!}{[2(k+l)]! \cdot l!} = 2^{-k} \cdot \frac{1}{\prod_{m=0}^{k-1}[2(l+m) + 1]} \cdot \frac{1}{(2l)!} \\
& = 4^{-k} \cdot
\frac{1}{\prod_{m=0}^{k-1}\left(l+ m + \frac{1}{2}\right)} \cdot \frac{1}{(2l)!} = 4^{-k} \cdot \frac{1}{(2l)! \left(l + \frac{1}{2}\right)_{k}}
\end{align*}
Hence 
\begin{align*}
& \sum_{k=0}^{\infty} (-1)^{k} \,
\frac{t^{2k} }{(2k)!}\, (A + B)^{k} = \sum_{l=0}^{\infty} (-1)^{l} \, t^{2l} 
\left[ \sum_{k=0}^{\infty} \frac{1}{(2l)! \left(l + \frac{1}{2}\right)_{k}} \cdot 
\frac{1}{k!} \, \left(- \, \frac{t^2}{4} \, . A\right)^{k}
 \right] \! B^{l} \\
& \sum_{l=0}^{\infty} (-1)^{l} \, \frac{t^{2l}}{(2l)!} 
\left[ \sum_{k=0}^{\infty} \frac{1}{ \left(l + \frac{1}{2}\right)_{k} \cdot k!} \cdot 
\left(- \, \frac{t^2}{4} \, . A\right)^{k} \right] \! B^{l} \, \, .
\end{align*}
By definition of the generalized hypergeometric function
$\leftidx{_0}{F}{_1}$, for every $l \in {\mathbb{N}}$, $z \in {\mathbb{C}}$
\begin{align} \label{generalizedhypergeometric}
\leftidx{_0}{F}{_1}\!\!\left(-; l + \frac{1}{2}; z\right)
= \sum_{k=0}^{\infty} \frac{z^k}{(l + \frac{1}{2})_k \cdot k!}
\, \, .
\end{align}
Hence 
\begin{align*}
& \sum_{k=0}^{\infty} (-1)^{k} \,
\frac{t^{2k} }{(2k)!}\, (A + B)^{k} 
= \sum_{l=0}^{\infty} (-1)^{l} \cdot \frac{t^{2l}}{(2l)!} \, . 
\left\{
\left[\leftidx{_0}{F}{_1}\!\!\left(-; l + \frac{1}{2}; - \, \frac{t^2}{4} \, . {\textrm id}_{\sigma(A)}\right)\right]\!\!(A)
\right\} \! B^l
\, \, .
\end{align*}
Furthermore,
\begin{align*}
& \sum_{k=0}^{\infty} (-1)^{k} \,
\frac{t^{2k+1} }{(2k+1)!}\, (A + B)^{k} =
\sum_{k=0}^{\infty} \sum_{l=0}^{k} (-1)^{k} \,
\frac{t^{2k+1} }{(2k+1)!}\, \binom{k}{l} A^{k - l} B^{l} \\
& = \sum_{l=0}^{\infty} \sum_{k=l}^{\infty} (-1)^{k} \,
\frac{t^{2k+1} }{(2k+1)!}\, \binom{k}{l} A^{k - l} B^{l} =
\sum_{l=0}^{\infty} \left[ \sum_{k=l}^{\infty} (-1)^{k} \, \binom{k}{l}
\frac{t^{2k+1} }{(2k+1)!}\, A^{k - l} \right] \! B^{l}  \\
& = \sum_{l=0}^{\infty} \left[ \sum_{k=0}^{\infty} (-1)^{k+l} \binom{k+l}{l} \,
\frac{t^{2(k+l)+1} }{[2(k+l)+1]!}\, A^{k} \right] \!  B^{l} \\
& =
\sum_{l=0}^{\infty} (-1)^{l} t^{2l+1} 
\left[ \sum_{k=0}^{\infty} (-1)^{k} \binom{k+l}{l} \,
\frac{t^{2k}}{[2(k+l)+1]!}\, A^{k} \right] \!  B^{l} \\
& = t \sum_{l=0}^{\infty} (-1)^{l}  \cdot t^{2l}
\left[ \sum_{k=0}^{\infty} \frac{(k+l)!}{[2(k+l)+1]! \cdot l!} \cdot \frac{1}{k!} \, .
\left(- t^2 A\right)^{k} \right] \!  B^{l} \, \, .
\end{align*}
Since for every $k, l \in {\mathbb{N}}$ 
\begin{align*}
& \frac{(k+l)!}{[2(k+l)+1]! \cdot l!} =
\frac{1}{2} \cdot \frac{1}{k+ l + \frac{1}{2}}
\cdot \frac{(k+l)!}{[2(k+l)]! \cdot l!} \\
& = \frac{1}{2} \cdot \frac{1}{k+ l + \frac{1}{2}} \cdot 4^{-k} \cdot \frac{1}{(2l)! \left(l + \frac{1}{2}\right)_{k}} = 
\frac{1}{2} \cdot 4^{-k} \cdot \frac{1}{(2l)! \left(l + \frac{1}{2}\right) \left(l + \frac{3}{2}\right)_{k}} \\
& = 4^{-k} \cdot \frac{1}{(2l+1)! \left(l + \frac{3}{2}\right)_{k}}
\, \, , 
\end{align*}
we conclude that 
\begin{align*}
& \sum_{k=0}^{\infty} (-1)^{k} \,
\frac{t^{2k+1} }{(2k+1)!}\, (A + B)^{k} \\
& = t \sum_{l=0}^{\infty} (-1)^{l}  \cdot t^{2l}
\left[ \sum_{k=0}^{\infty}
\frac{1}{(2l+1)! \left(l + \frac{3}{2}\right)_{k}}
\cdot \frac{1}{k!} \, .
\left(- \, \frac{t^2}{4} \, . A\right)^{k} \right] \! B^{l} \\
& = \sum_{l=0}^{\infty} (-1)^{l}  \cdot \frac{t^{2l+1}}{(2l+1)!}
\left[\sum_{k=0}^{\infty}
\frac{1}{\left(l + \frac{3}{2}\right)_{k} \cdot k!}
\, .
\left(- \, \frac{t^2}{4} \, . A\right)^{k} \right] \! B^{l} \\
& = \sum_{l=0}^{\infty} (-1)^{l} \cdot \frac{t^{2l+1}}{(2l+1)!} \, . 
\left\{
\left[\leftidx{_0}{F}{_1}\!\!\left(-; l + \frac{3}{2}; - \, \frac{t^2}{4} \, . {\textrm id}_{\sigma(A)}\right)\right]\!\!(A)
\right\} \! B^l
\, \, .
\end{align*}
\end{proof}

The following lemma gives a connection between generalized hypergeometric
and spherical Bessel functions.

\begin{lem} \label{additiontheorem1}
For every $k \in {\mathbb{N}}$ and $x > 0$
\begin{align*}
\frac{x^{2k}}{(2k)!} \, . \,
\leftidx{_0}{F}{_1}\!\!\left(-; k + \frac{1}{2}; - \, \frac{x^2}{4}  \right) & = \frac{1}{2^k k!} \, x^{k+1} j_{k-1}(x) \, \, , \\
\frac{x^{2k+1}}{(2k+1)!} \cdot \leftidx{_0}{F}{_1}\left(-;k + \frac{3}{2}, - \frac{x^2}{4} \right) & = \frac{1}{2^k k!} \, x^{k+1} j_k(|x|) \, \, , 
\end{align*}
where the spherical Bessel functions 
$j_0,j_1,\dots$ are defined as in \cite{olverEtAl2010_book} and
\begin{equation*}
j_{-1}(x) := \frac{\cos(x)}{x}, \quad x>0.
\end{equation*}
\end{lem}

\begin{proof}
We note that for every $\nu \in (0,\infty)$, $k \in {\mathbb{N}}$, and $x > 0$
\begin{align*}
J_{\nu}(x) & :=
 \left(\frac{x}{2} \right)^{\nu} \sum_{k=0}^{\infty} \frac{(-1)^k}{k! \, \Gamma(\nu + k + 1)} \left(\frac{x^2}{4}\right)^{k} = \frac{1}{\Gamma(\nu + 1)} \cdot \left(\frac{x}{2} \right)^{\nu} \sum_{k=0}^{\infty} \frac{(-1)^k}{k! \, \frac{\Gamma(\nu + k + 1)}{\Gamma(\nu + 1)}} \left(\frac{x^2}{4}\right)^{k} 
 \\
 & \phantom{:}= \frac{1}{\Gamma(\nu + 1)} \cdot \left(\frac{x}{2} \right)^{\nu} \sum_{k=0}^{\infty} \frac{(-1)^k}{k! \, (\nu + 1)_k} \left(\frac{x^2}{4}\right)^{k} 
= \frac{1}{\Gamma(\nu + 1)} \cdot \left( \frac{x}{2} \right)^{\nu} \cdot \leftidx{_0}{F}{_1}(-;\nu + 1,- x^2/4) \, \,.
\end{align*}

\begin{align*}
j_k(x) & := \sqrt{\frac{\pi}{2 x}} \, J_{k+\frac{1}{2}}(x) =
\sqrt{\frac{\pi}{2 x}} \, \frac{1}{\Gamma(k + \frac{3}{2})} \cdot \left( \frac{x}{2} \right)^{k + \frac{1}{2}} \cdot \leftidx{_0}{F}{_1}(-;k + \frac{3}{2}, -x^2/4) \\
& \phantom{:}= \frac{\sqrt{\pi}}{2 \Gamma(k + \frac{3}{2})} \cdot \left( \frac{x}{2} \right)^{k} \cdot \leftidx{_0}{F}{_1}(-;k + \frac{3}{2}, -x^2/4) \, \, .
\end{align*}
Hence for every $k \in {\mathbb{N}}$, $x > 0$
\begin{align*}
\leftidx{_0}{F}{_1}(-;k + \frac{3}{2}, -x^2/4) = 
\frac{2 \Gamma(k + \frac{3}{2})}{\sqrt{\pi}} \left(\frac{x}{2} \right)^{-k} j_k(x) =
2^{k+1} \left(\frac{1}{2}\right)_{\!k+1} x^{-k} j_k(x)
\end{align*}
as well as
\begin{align*}
& \frac{x^{2k+1}}{(2k+1)!} \, \leftidx{_0}{F}{_1}(-;k + \frac{3}{2}, -x^2/4)
= \frac{x^{2k+1}}{(2k+1)!} \, 2^{k+1} \left(\frac{1}{2}\right)_{\!k+1} x^{-k} j_k(x) \\
& = \frac{x^{2k+1}}{(2k+1)!} \, 2^{k+1} \,  2^{-(k+1)} \, \frac{(2k + 2)!}{2^{k+1} (k+1)!} x^{-k} j_k(x) = \frac{(2 k + 2)}{2^{k+1}(k+1)!} \, x^{k+1} j_k(x) \\
& = \frac{1}{2^k k!} \, x^{k+1} j_k(x) \, \, .
\end{align*}
Furthermore, for every $k \in {\mathbb{N}}^{*}$, $x > 0$
\begin{equation} \label{auxiliaryrelation}
\frac{x^{2k}}{(2k)!} \, . \,
\leftidx{_0}{F}{_1}(-; k + \frac{1}{2}; - x^2/4) = \frac{x}{2k} \, \frac{1}{2^{k-1} (k-1)!} \, x^{k} j_{k-1}(x)
= \frac{1}{2^k k!} \, x^{k+1} j_{k-1}(x) \, \, .
\end{equation}
Since for $x > 0$
\begin{align*}
& \leftidx{_0}{F}{_1}(-;\frac{1}{2}, -x^2/4) = \sum_{k=0}^{\infty} \frac{1}{\left(\frac{1}{2}\right)_{k}\cdot k!} \cdot \left(- \, \frac{x^2}{4}\right)^{k} =
\sum_{k=0}^{\infty} \frac{1}{2^{-k} \cdot \frac{(2k)!}{2^k \cdot k!} \cdot k!} \cdot \left(- \, \frac{x^2}{4}\right)^{k} \\
& = \sum_{k=0}^{\infty} \frac{4^k}{(2k)!} \cdot \left(- \, \frac{x^2}{4}\right)^{k} = \sum_{k=0}^{\infty} (-1)^k \, \frac{x^{2k}}{(2k)!}
= \cos(x) \, \, , 
\end{align*}
the equality (\ref{auxiliaryrelation}) is true also for $k=0$, if 
we define 
\begin{equation*}
j_{-1}(x) := \frac{\cos(x)}{x} \, \, . 
\end{equation*}
\end{proof}

Eventually, we have a representation involving two commuting
operators, with one of the operators being a multiple of the identity
and a general $C$ which is not necessarily a convolution operator.

\begin{thm} \label{additiontheorem2}
Let $(X,\braket{\,|\,})$ be a non-trivial complex 
Hilbert space, $\sqrt{\phantom{ij}}$
the complex square-root function, with domain ${\mathbb{C}}
\setminus ((- \infty,0] \times \{0\})$. $c >0$, $C \in L(X,X)$ self-adjoint and $\sigma(c-C) \subset {\mathbb{R}}$ the 
(non-empty, compact) spectrum of $c - C$.
Then for every $t \in {\mathbb{R}}$
\begin{align*}
& \left[\overline{\cos \left(t \sqrt{\phantom{ij}} \right)}\,
\bigg|_{\sigma(c-C)}\right]\!(c - C) = \sum_{k=0}^{\infty} 
\,  \frac{1}{2^k k!} \, (\sqrt{c t^2} \,)^{k+1} j_{k-1}(\sqrt{c t^2} \, ) \! \left(\frac{1}{c} . C \right)^{\!\!k} 
\, \, , \\
& \left[\, \overline{\frac{\sin \left(t \sqrt{\phantom{ij}} \right)}{\sqrt{\phantom{ij}}}} \, \bigg|_{\sigma(c - C)}\right]\!(c - C)  = t \sum_{k=0}^{\infty} 
\,  \frac{1}{2^k k!} \, (\sqrt{c t^2} \,)^{k} j_{k}(\sqrt{c t^2} \,) \! \left(\frac{1}{c} . C \right)^{\!\!k} \, \, ,
\end{align*} 
where the spherical Bessel functions 
$j_0,j_1,\dots$ are defined as in \cite{olverEtAl2010_book} and
\begin{equation*}
j_{-1}(x) := \frac{\cos(x)}{x}, \quad x>0
\end{equation*}
and the members of the sums are defined for $t = 0$ by continuous extension. 
\end{thm}

\begin{proof}
Direct consequence of Theorem \ref{additiontheorem} and Lemma \ref{additiontheorem1}.
\end{proof}

We provide an error estimate of the previous representation.

\begin{cor} \label{errorestimates} {\bf (Error Estimates)}
Let $(X,\braket{\,|\,}),\sqrt{\phantom{ij}}, c, C, \sigma(c-C),j_{-1},j_0,j_1,\dots$ as in Theorem~\ref{additiontheorem2} and 
$N \in {\mathbb{N}}$.
Then for every $t \in {\mathbb{R}}$
\begin{align*}
& \left\|
\left[\overline{\cos \left(t \sqrt{\phantom{ij}} \right)} \,
\bigg|_{\sigma(c-C)}\right]\!(c - C) - \sum_{k=0}^{N} 
\,  \frac{1}{2^k k!} \, (\sqrt{c t^2} \,)^{k+1} j_{k-1}(\sqrt{c t^2} \, ) \! \left(\frac{1}{c} . C \right)^{\!\!k} \right\| \\
& \leqslant \frac{\pi}{N!} \, \min\left\{ 1,\left(\!\frac{t^2 \|C\|}{4}\!\right)^{\!\!N+1}\right\} e^{\,t^2 \|C\|/4} \, \, , \\
& \left\| \left[\, \overline{\frac{\sin \left(t \sqrt{\phantom{ij}} \right)}{\sqrt{\phantom{ij}}}} \, \bigg|_{\sigma(c - C)}\right]\!(c - C)  - t \sum_{k=0}^{\infty} 
\,  \frac{1}{2^k k!} \, (\sqrt{c t^2} \,)^{k} j_{k}(\sqrt{c t^2} \,) \! \left(\frac{1}{c} . C \right)^{\!\!k} \right\| \\
& \leqslant \frac{\pi}{2 (N+1)!} \, |t|  \, \min\left\{ 1,\left(\!\frac{t^2 \|C\|}{4}\!\right)^{\!\!N+1}\right\} e^{\,t^2 \|C\|/4} \, \, .
\end{align*} 
\end{cor}

\begin{proof}
As a consequence of Theorem~\ref{additiontheorem2}, for $t \in {\mathbb{R}}$
\begin{align*}
& \left\|
\left[\overline{\cos \left(t \sqrt{\phantom{ij}} \right)}\,
\bigg|_{\sigma(c-C)}\right]\!(c - C) - \sum_{k=0}^{N} 
\,  \frac{1}{2^k k!} \, (\sqrt{c t^2} \,)^{k+1} j_{k-1}(\sqrt{c t^2} \, ) \! \left(\frac{1}{c} . C \right)^{\!\!k} \right\| \\
& \leqslant \sum_{k=N+1}^{\infty} 
\,  \frac{1}{2^k k!} \, (\sqrt{c t^2} \,)^{k+1} |\,j_{k-1}(\sqrt{c t^2} \, )| \left(\!\frac{\|C\|}{c}\!\right)^{\!\!k} \\
& \leqslant \sum_{k=N+1}^{\infty} 
\,  \frac{1}{2^k k!} \, (\sqrt{c t^2} \,)^{k+1} \, 
\pi \, \frac{(\sqrt{c t^2} \,)^{k-1}}{2^k (k-1)!}
\left(\!\frac{\|C\|}{c}\!\right)^{\!\!k} \\
& = \pi \sum_{k=N+1}^{\infty} 
\,  \frac{1}{k!(k-1)!}
\left(\!\frac{t^2 \|C\|}{4}\!\right)^{\!\!k} 
\leqslant \frac{\pi}{N!} \sum_{k=N+1}^{\infty} 
\,  \frac{1}{k!}
\left(\!\frac{t^2 \|C\|}{4}\!\right)^{\!\!k} \, \, ,
\end{align*}
where the integral representation DLMF~10.54.1 of \cite{olverEtAl2010_book} 
({\tt http://dlmf.nist.gov/10.54}) for spherical 
Bessel functions has been used. Since 
\begin{align*}
& \sum_{k=N+1}^{\infty} 
\,  \frac{1}{k!}
\left(\!\frac{t^2 \|C\|}{4}\!\right)^{\!\!k} =
\left(\!\frac{t^2 \|C\|}{4}\!\right)^{\!\!N+1} \sum_{k=N+1}^{\infty} 
\,  \frac{1}{k!}
\left(\!\frac{t^2 \|C\|}{4}\!\right)^{\!\!k - N - 1} \\
& \leqslant \left(\!\frac{t^2 \|C\|}{4}\!\right)^{\!\!N+1}
\sum_{k=N+1}^{\infty} 
\,  \frac{1}{(k - N - 1)!}
\left(\!\frac{t^2 \|C\|}{4}\!\right)^{\!\!k - N - 1}
\leqslant \left(\!\frac{t^2 \|C\|}{4}\!\right)^{\!\!N+1}
e^{\,t^2 \|C\|/4} \, \, , 
\end{align*}
this implies that 
\begin{align*}
& \left\|
\left[\overline{\cos \left(t \sqrt{\phantom{ij}} \right)}\,
\bigg|_{\sigma(c-C)}\right]\!(c - C) - \sum_{k=0}^{N} 
\,  \frac{1}{2^k k!} \, (\sqrt{c t^2} \,)^{k+1} j_{k-1}(\sqrt{c t^2} \, ) \! \left(\frac{1}{c} . C \right)^{\!\!k} \right\| \\
& \leqslant \frac{\pi}{N!} \, \min\left\{ 1,\left(\!\frac{t^2 \|C\|}{4}\!\right)^{\!\!N+1}\right\} e^{\,t^2 \|C\|/4} \, \, .
\end{align*}
Furthermore, 
\begin{align*}
& \left\|\left[\, \overline{\frac{\sin \left(t \sqrt{\phantom{ij}} \right)}{\sqrt{\phantom{ij}}}} \, \bigg|_{\sigma(c - C)}\right]\!(c - C) - t \sum_{k=0}^{N} 
\,  \frac{1}{2^k k!} \, (\sqrt{c t^2} \,)^{k} j_{k}(\sqrt{c t^2} \, ) \! \left(\frac{1}{c} . C \right)^{\!\!k} \right\| \\
& \leqslant |t| \sum_{k=N+1}^{\infty} 
\,  \frac{1}{2^k k!} \, (\sqrt{c t^2} \,)^{k} |\,j_{k}(\sqrt{c t^2} \, )| \left(\!\frac{\|C\|}{c}\!\right)^{\!\!k} \\
& \leqslant |t| \sum_{k=N+1}^{\infty} 
\,  \frac{1}{2^k k!} \, (\sqrt{c t^2} \,)^{k} \, 
\pi \, \frac{(\sqrt{c t^2} \,)^{k}}{2^{k+1} k!}
\left(\!\frac{\|C\|}{c}\!\right)^{\!\!k} \\
& = \frac{\pi}{2} \, |t| \sum_{k=N+1}^{\infty} 
\,  \frac{1}{(k!)^2}
\left(\!\frac{t^2 \|C\|}{4}\!\right)^{\!\!k} 
\leqslant \frac{\pi}{2 (N+1)!} \, |t| \sum_{k=N+1}^{\infty} 
\,  \frac{1}{k!}
\left(\!\frac{t^2 \|C\|}{4}\!\right)^{\!\!k} \\
& \leqslant \frac{\pi}{2 (N+1)!} \, |t|  \, \min\left\{ 1,\left(\!\frac{t^2 \|C\|}{4}\!\right)^{\!\!N+1}\right\} e^{\,t^2 \|C\|/4} \, \, .
\end{align*}
\end{proof}

Application of Theorem \ref{additiontheorem2} to the special case of
$A_C$ gives the following.

\begin{thm} \label{besselrepresentationtheorem}
Let $n \in {\mathbb{N}}^{*}$, $\rho > 0$, 
$C \in L^1({\mathbb{R}}^n)$ be even such that 
\begin{equation*}
c := \int_{{\mathbb{R}}^n} C \, dv^n > 0  
\end{equation*}
and $A_C$ as in Lemma~\ref{governingoperator}. Then for $t \in {\mathbb{R}}$
\begin{align} \label{besselrepresentation}
& \left[\overline{\cos \left(t \sqrt{\phantom{ij}} \right)}\,
\bigg|_{\sigma(A_C)}\right]\!(A_C) f = \sum_{k=0}^{\infty} \frac{1}{2^k k!} ( \sqrt{c t^2/ \rho}\,)^{k+1} j_{k-1}
( \sqrt{c t^2 / \rho} \,) \, c^{-k} . C^k * f \, \, , \nonumber \\
& 
\left[\, \overline{\frac{\sin \left(t \sqrt{\phantom{ij}} \right)}{\sqrt{\phantom{ij}}}} \, \bigg|_{\sigma(A_C)}\right]\!(A_C)
f = t \sum_{k=0}^{\infty} \frac{1}{2^k k!} ( \sqrt{c t^2/ \rho}\,)^{k} j_{k}
( \sqrt{c t^2 / \rho} \,) \, c^{-k} . C^k * f \, \, , 
\end{align}
for every $f \in L^2_{\mathbb{C}}({\mathbb{R}}^n)$, where the spherical Bessel functions 
$j_0,j_1,\dots$ are defined as in \cite{olverEtAl2010_book}, 
\begin{equation*}
j_{-1}(z) := \frac{\cos(z)}{z}, \quad z \in {\mathbb{C}}^{*},
\end{equation*}
and the members of the sums are defined for $t = 0$ by continuous extension.  
\end{thm}

\begin{proof}
The statement is a direct consequence of Theorems~\ref{governingoperator}
and \ref{additiontheorem2}.
\end{proof}

\section{Examples} \label{sec:examples}
\begin{figure}[t]
\centering 
\subfigure[Generalized solution $u$ to the classical (local) wave equation with 
initial data $u(0,x) = 1/(1+x^2)$ and $(\partial u / \partial
t)(0,x)=0,~x \in {\mathbb{R}}$.]{
\scalebox{0.465}{\includegraphics{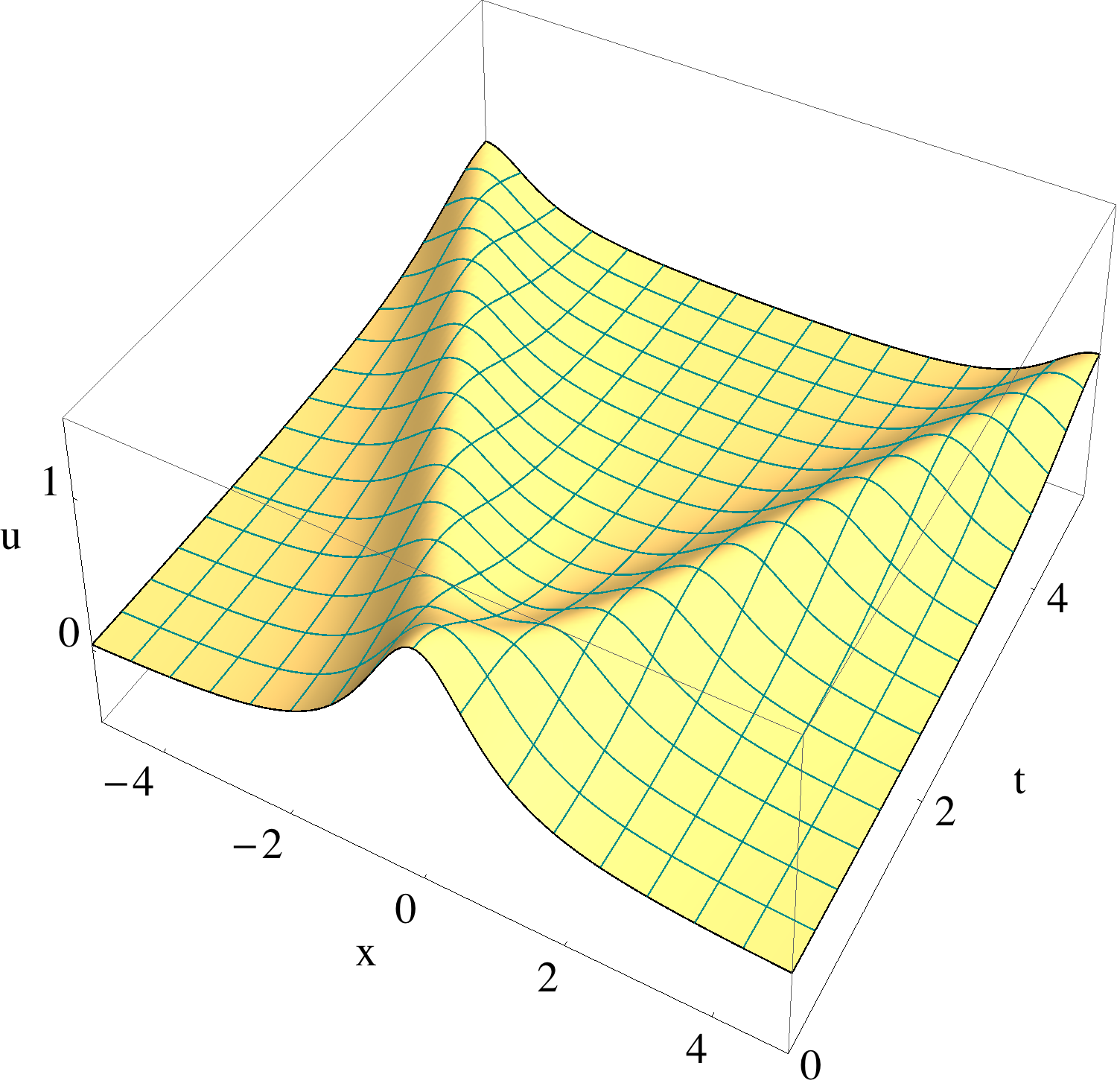}}
\label{fig0}
}
\hspace{.35cm}
\subfigure[Solution $u$ to the nonlocal wave equation 
with initial data $u(0,x) = f(x)$ 
and $(\partial u / \partial
t)(0,x)=0,~x \in {\mathbb{R}}$
in Example~\ref{gaussians}.]{
\scalebox{0.465}{\includegraphics{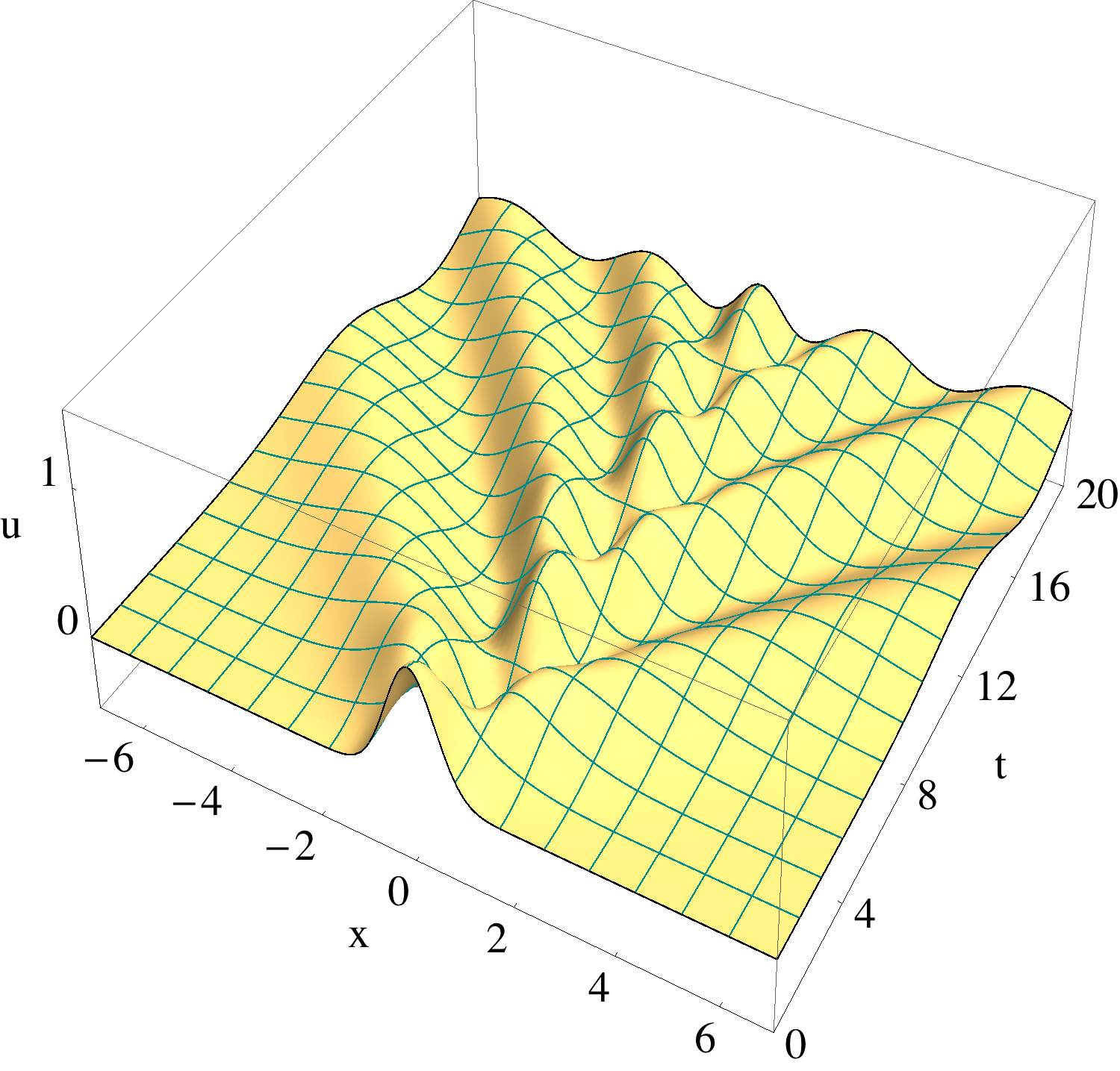}}
\label{fig1}
}
\\
\subfigure[Generalized solution $u$ to the classical (local) wave equation with 
initial data $u(0,x) = 0$ and $(\partial u / \partial t)(0,x)= 1/(1+x^2), 
~x \in {\mathbb{R}}$.]{
\scalebox{0.465}{\includegraphics{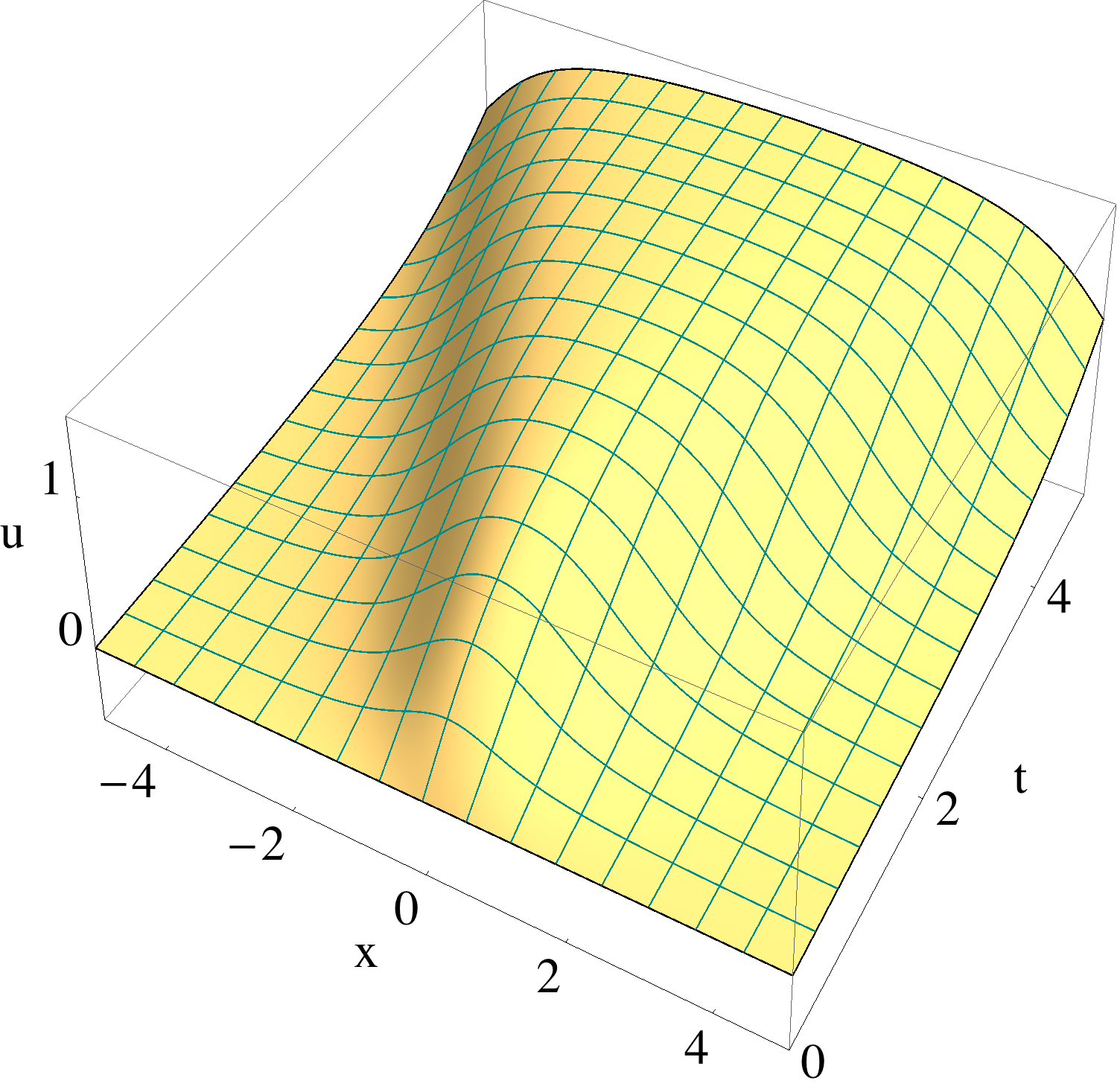}}
\label{fig2}
}
\hspace{.35cm}
\subfigure[Solution $u$ to the nonlocal wave equation 
with initial data $u(0,x) = 0$ 
and $(\partial u / \partial
t)(0,x)=f(x),~x \in {\mathbb{R}}$
in Example~\ref{gaussians}.]{
\scalebox{0.465}{\includegraphics{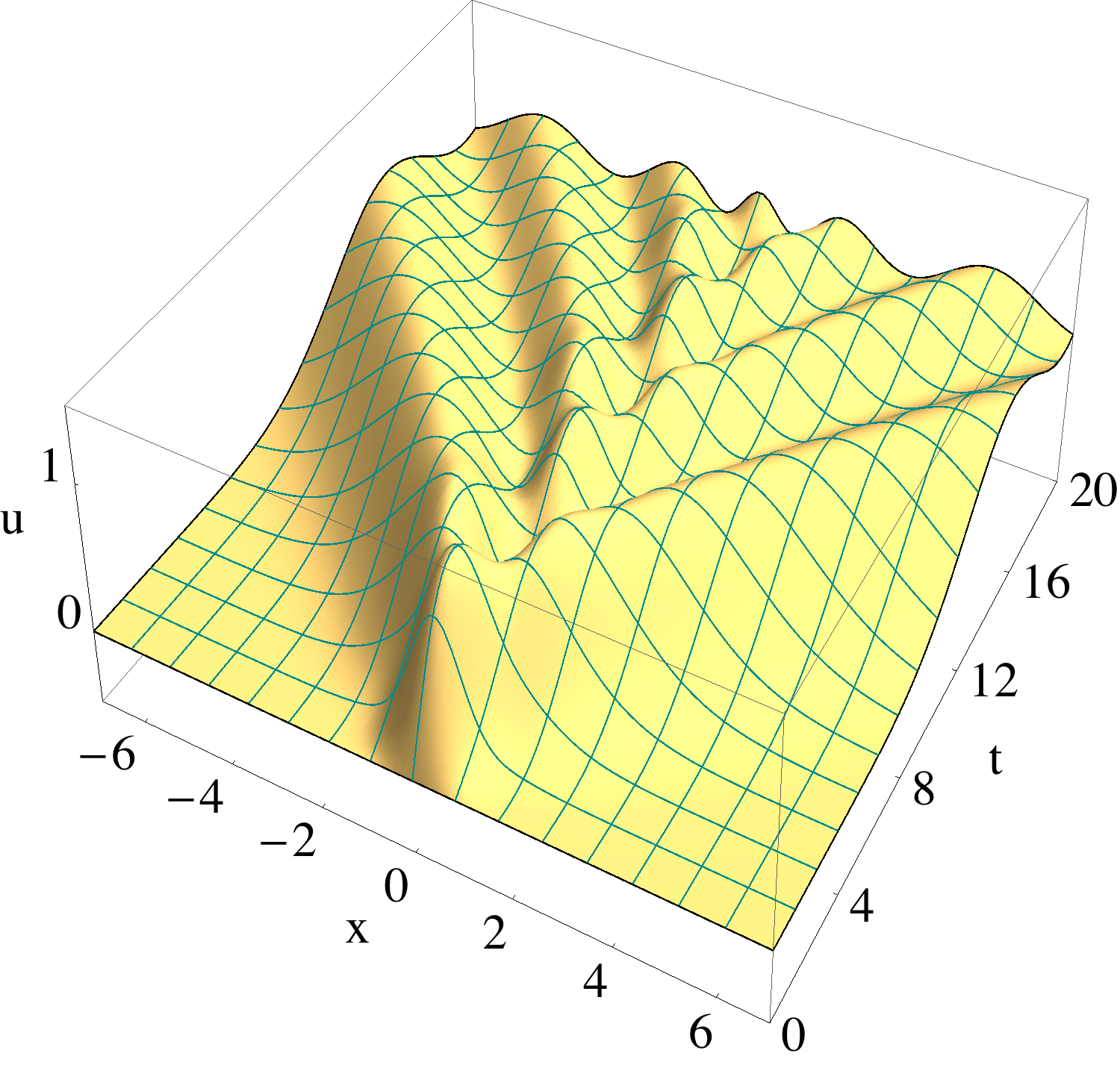}}
\label{fig3}
}
\caption{Evolution of the local and nonlocal wave equation solutions with
vanishing initial velocity ((a) and (b)) and vanishing initial displacement
((c) and (d)). For (a) and (c), we use $\rho = E = 1$, $b = 0$,
values in (\ref{classicalelasticity}). For (b) and (d), 
we use $c = a = 1$, $\rho = 1$, $\sigma = 1$, $\sigma_d = 1/2$
values in Example~\ref{gaussians}.}
\label{fig:comparisonContinuous}
\end{figure}

We apply the apparatus we have constructed of the previous section on
an example that involves a micromodulus and input function both
of which are normal distributions with mean value zero
and standard deviation $\sigma$ and $\sigma_d$, respectively.

\begin{ex} \label{gaussians}
For $\rho,\sigma, \sigma_{d}, a > 0$, we define $C_{\sigma} 
\in L^1({\mathbb{R}})$ and $f \in L^2_{\mathbb{C}}({\mathbb{R}})$
by 
\begin{align*} 
& C_{\sigma} := 
\frac{a}{\sqrt{2 \pi} \, \sigma} \, e^{- [1 / (2 \sigma^2)] . {\mathrm{id}}_{\mathbb{R}}^2}
\, \, , \, \, 
f := \frac{1}{\sqrt{2 \pi}\sigma_d} \, e^{- [1 / (2 \sigma_d^2)] . {\mathrm{id}}_{\mathbb{R}}^2} \, \, .
\end{align*}
Then for $k \in {\mathbb{N}}^{*}$
\begin{align*}
& F_1 C_{\sigma} =  a e^{- (\sigma^2/2). {\mathrm{id}}_{\mathbb{R}}^2} \, \, , \, \, 
F_1 C_{\sigma}^k = (F_1 C_{\sigma})^k =   
a^k e^{- k (\sigma^2/2). {\mathrm{id}}_{\mathbb{R}}^2} \, \, , \\
& F_2 (C_{\sigma}^k * f) = (F_1 C_{\sigma}^{k}) \cdot F_2 f = \frac{a^k}{\sqrt{2 \pi}} \,
e^{- k (\sigma^2/2). {\mathrm{id}}_{\mathbb{R}}^2} \cdot e^{- (\sigma_d^2/2). {\mathrm{id}}_{\mathbb{R}}^2} \\
& 
= \frac{a^k}{\sqrt{2 \pi}} \, e^{- [(k \sigma^2 + \sigma_d^2)/2]. {\mathrm{id}}_{\mathbb{R}}^2}  =  
\frac{1}{\sqrt{2 \pi}} \, F_1 \frac{a^k}{\sqrt{2 \pi} \sqrt{k \sigma^2 + \sigma_d^2}} \, e^{- \{1 / [2(k \sigma^2 + \sigma_d^2)]\} . {\mathrm{id}}_{\mathbb{R}}^2} \\
& = F_2 \frac{a^k}{\sqrt{2 \pi} \sqrt{k \sigma^2 + \sigma_d^2}} \, e^{- \{1 / [2(k \sigma^2 + \sigma_d^2)]\} . {\mathrm{id}}_{\mathbb{R}}^2} 
\end{align*}
and hence
\begin{equation*}
C_{\sigma}^k * f = \frac{a^k}{\sqrt{2 \pi} \sqrt{k \sigma^2 + \sigma_d^2}} \, e^{- \{1 / [2(k \sigma^2 + \sigma_d^2)]\} . {\mathrm{id}}_{\mathbb{R}}^2} \, \, .
\end{equation*}
Since
\begin{equation*}
c = \int_{\mathbb{R}} C_{\sigma} \, dv^1 = a > 0 \, \, , 
\end{equation*}
we conclude from Theorem~\ref{besselrepresentationtheorem} that for $t \in {\mathbb{R}}$
\begin{eqnarray} \label{specialbesselrepresentation}
& \left[\overline{\cos \left(t \sqrt{\phantom{ij}} \right)}\,
\bigg|_{\sigma(A_C)}\right]\!(A_C) f \nonumber \\
& = \sum_{k=0}^{\infty} \frac{1}{2^k k!} ( \sqrt{a t^2/ \rho}\,)^{k+1} j_{k-1}
( \sqrt{a t^2 / \rho} \,) \, 
\frac{1}{\sqrt{2 \pi} \sqrt{k \sigma^2 + \sigma_d^2}} \, e^{- \{1 / [2(k \sigma^2 + \sigma_d^2)]\} . {\mathrm{id}}_{\mathbb{R}}^2} \, \, , \nonumber \\
& 
\left[\, \overline{\frac{\sin \left(t \sqrt{\phantom{ij}} \right)}{\sqrt{\phantom{ij}}}} \, \bigg|_{\sigma(A_C)}\right]\!(A_C)
f \nonumber \\
& = t \sum_{k=0}^{\infty} \frac{1}{2^k k!} ( \sqrt{a t^2/ \rho}\,)^{k} j_{k}
( \sqrt{a t^2 / \rho} \,) \, 
\frac{1}{\sqrt{2 \pi} \sqrt{k \sigma^2 + \sigma_d^2}} \, e^{- \{1 / [2(k \sigma^2 + \sigma_d^2)]\} . {\mathrm{id}}_{\mathbb{R}}^2} \, \, , 
\end{eqnarray}
where $A_C$ is as in Lemma~\ref{governingoperator}. 
\end{ex}

\begin{figure}[t]
\centering
\subfigure[Generalized solution $u$ to the classical (local) wave equation with 
initial data $u(0,x) = e^{- {\textrm id}_{\mathbb{R}}} \chi_{_{[0,\infty)}}(x)$ 
and 
$(\partial u / \partial t)(0,x)=0, x \in {\mathbb{R}}$.]{
\scalebox{0.465}{\includegraphics{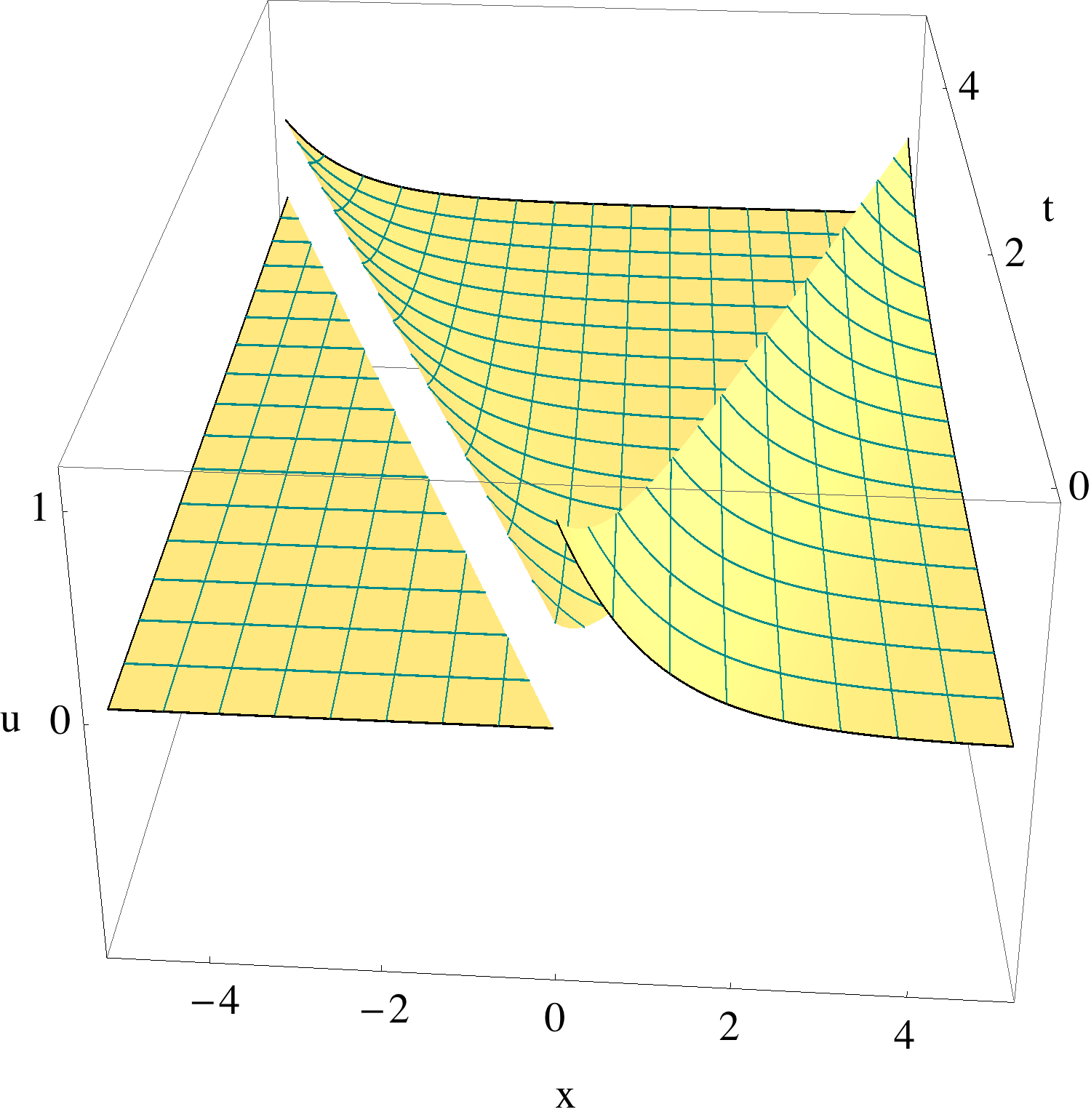}}
\label{fig4}
}
\hspace{.35cm}
\subfigure[Solution $u$ to the nonlocal wave equation 
with initial data 
$u(0,x) = e^{- {\textrm id}_{\mathbb{R}}} \chi_{_{[0,\infty)}}(x)$ 
and $(\partial u / \partial t)(0,x)=0,~x \in {\mathbb{R}}$.]{
\scalebox{0.465}{\includegraphics{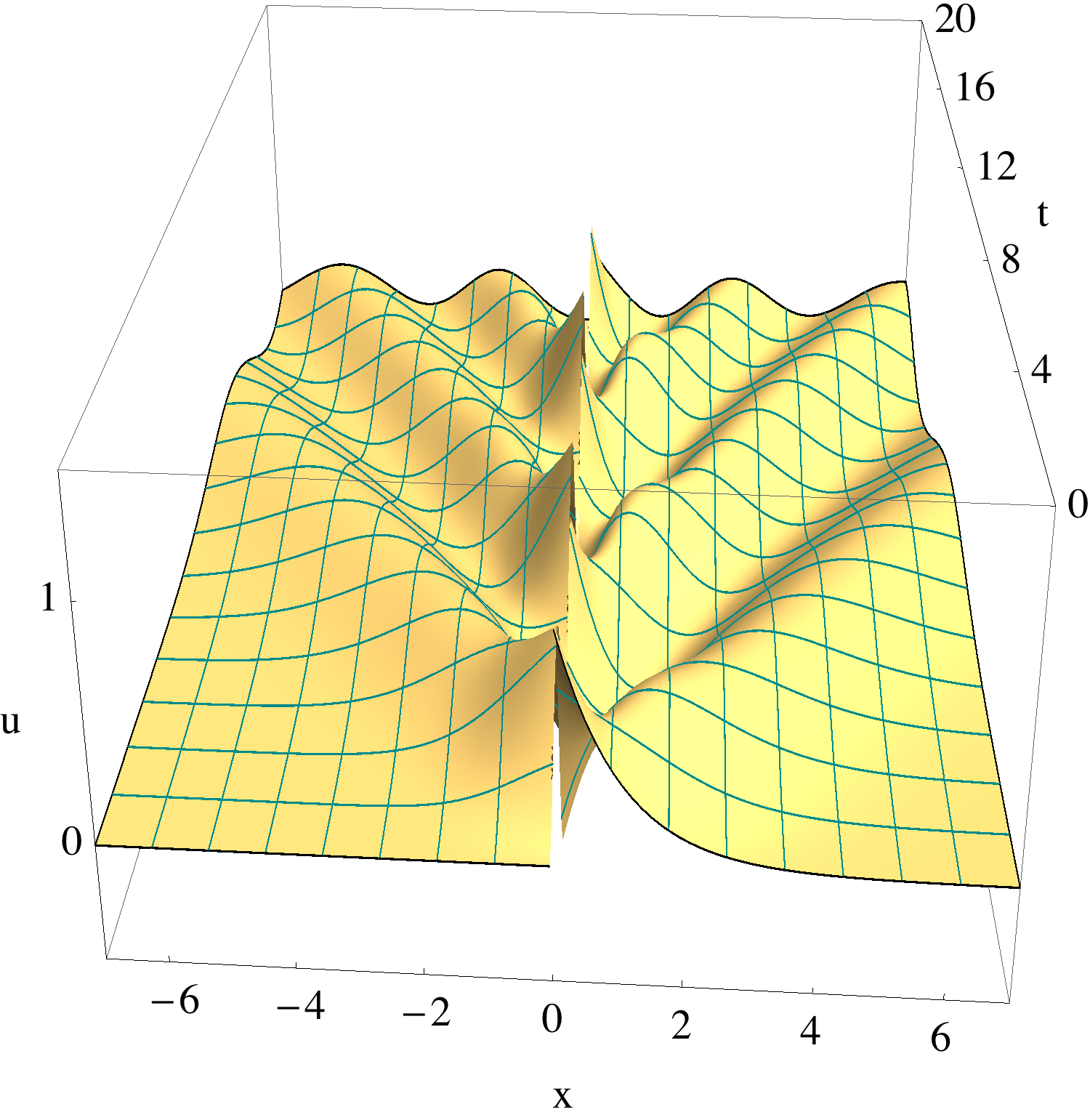}}
\label{fig5}
}
\\
\subfigure[Generalized solution $u$ to the classical (local) wave equation with 
initial data $u(0,x) = 0$ and $(\partial u / \partial t)(0,x) =  
e^{- {\textrm id}_{\mathbb{R}}} \chi_{_{[0,\infty)}}(x), ~x \in {\mathbb{R}}$.]{
\scalebox{0.465}{\includegraphics{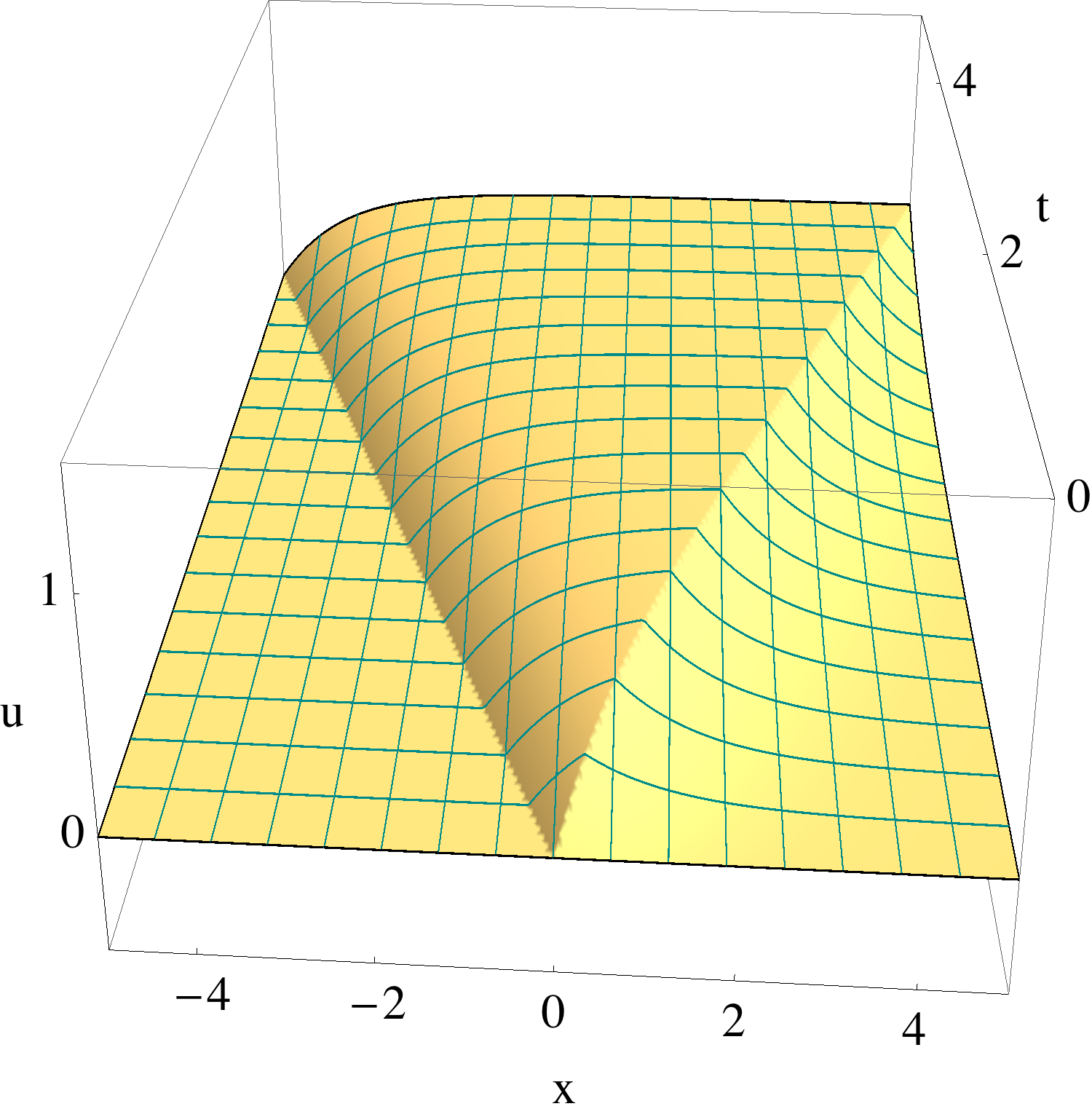}}
\label{fig6}
}
\hspace{.35cm}
\subfigure[Solution $u$ to the nonlocal wave equation 
with initial data $u(0,x) = 0$ 
and $(\partial u / \partial
t)(0,x)=f(x),~x \in {\mathbb{R}}$
in Example~\ref{gaussians1}.]{
\scalebox{0.465}{\includegraphics{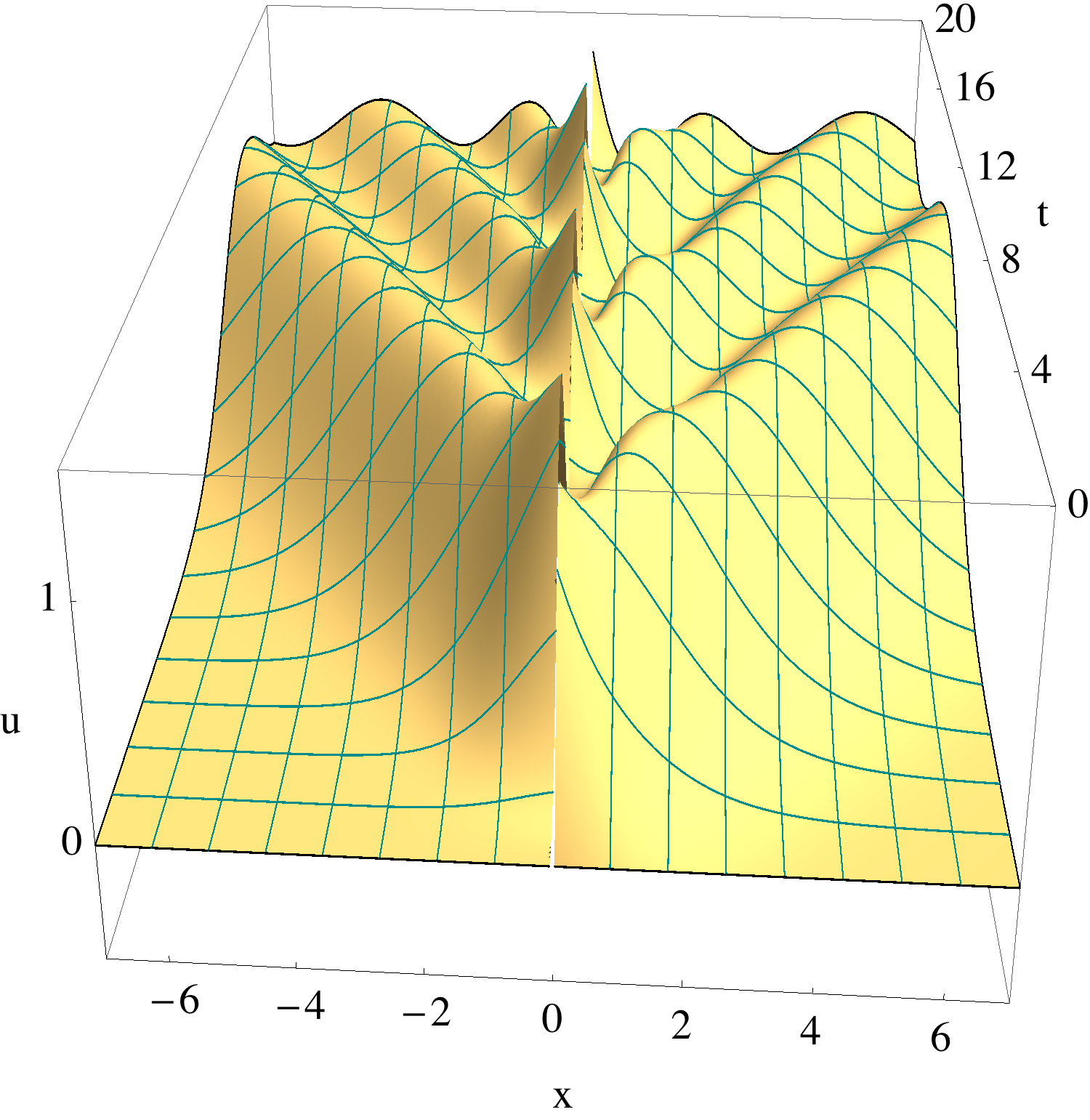}}
\label{fig7}
}
\caption{Evolution of the local and nonlocal wave equation solutions with
discontinuous initial displacement ((a) and (b)) and discontinuous 
initial velocity 
((c) and (d)). For (a) and (c), we use $\rho = E = 1$, $b = 0$,
values in (\ref{classicalelasticity}). For (b) and (d), 
we use $c = a = 1$, $\rho = 1$, $\sigma = 1$, and $b=\epsilon=1$
values in Example~\ref{gaussians1}.}
\label{fig:comparisonDiscontinuous}
\end{figure}

We depict and compare the solutions of the classical and nonlocal wave
equations in Figures~\ref{fig:comparisonContinuous} and
\ref{fig:comparisonDiscontinuous}.  In the classical case, as expected,
we observe the propagation of waves along characteristics; see
Figures~\ref{fig0} and \ref{fig2} for vanishing initial velocity and
displacement, respectively.  In the nonlocal case, we observe repeated
separation of waves and an oscillation at the center of the initial
pulse; see Figures~\ref{fig1} and \ref{fig3} for vanishing initial
velocity and displacement, respectively.

We study the propagation of discontinuity in the data for classical 
and nonlocal wave equations in the following example.

\begin{ex} \label{gaussians1}
As in the previous example, for $\rho,\sigma, a, b, \varepsilon > 0$, we define $C_{\sigma} 
\in L^1({\mathbb{R}})$ 
by 
\begin{align*} 
& C_{\sigma} := 
\frac{a}{\sqrt{2 \pi} \, \sigma} \, e^{- [1 / (2 \sigma^2)] . {\mathrm{id}}_{\mathbb{R}}^2}
 \, \, .
\end{align*}
Then 
\begin{equation*}
F_1 C_{\sigma} =  
a e^{- (\sigma^2/2). {\mathrm{id}}_{\mathbb{R}}^2} \, \, , 
\end{equation*}
and for $k \in {\mathbb{N}}^{*}$,
\begin{equation*}
F_1 C_{\sigma}^k = (F_1 C_{\sigma})^k =   
a^k e^{- k (\sigma^2/2). {\mathrm{id}}_{\mathbb{R}}^2} =
F_1 \frac{a^k}{\sqrt{2 \pi k} \, \sigma} \, e^{- [1 / (2  k \sigma^2)] . {\mathrm{id}}_{\mathbb{R}}^2}
\, \, , 
\end{equation*}
and hence
\begin{equation*}
C_{\sigma}^k = \frac{a^k}{\sqrt{2 \pi} \sqrt{k \sigma^2}} \, e^{- [1 / (2k \sigma^2) ] . {\mathrm{id}}_{\mathbb{R}}^2} \, \, .
\end{equation*}
Furthermore, we define $f \in L^2_{\mathbb{C}}({\mathbb{R}})$
by 
\begin{align*}
f := b \, e^{- \varepsilon . {\mathrm{id}}_{\mathbb{R}}} \cdot \chi_{_{[0,\infty)}} \, \, .
\end{align*}
Then for $x \in {\mathbb{R}}$, 
\begin{align*}
& (C_{\sigma}^k * f)(x) = \frac{a^k b}{\sqrt{2 \pi} \sqrt{k \sigma^2}} \, \int_{0}^{\infty} e^{- [(x -y)^2 / (2k \sigma^2)]} \cdot e^{- \varepsilon y} \, dy \\
& =
\frac{2 a^k b}{\pi} \, e^{\left(
\frac{\varepsilon \sigma}{2} \sqrt{2k} \right)^2}
\, e^{- \varepsilon x} \, \mathrm{erfc}\!\left(
\frac{\varepsilon \sigma}{2}\sqrt{2k} - \frac{x}{\sigma \sqrt{2k}} 
\right) 
 \, \, ,
\end{align*}
where $\mathrm{erfc}$ denotes the error function defined according to
DLMF \cite{olverEtAl2010_book}. 
We note for $x \in {\mathbb{R}}$ that 
\begin{equation*}
\lim_{k \rightarrow 0} \frac{a^k b}{2} \, e^{\left(
\frac{\varepsilon \sigma}{2} \sqrt{2k} \right)^2}
\, e^{- \varepsilon x} \, \mathrm{erfc}\!\left(
\frac{\varepsilon \sigma}{2}\sqrt{2k} - \frac{x}{\sigma \sqrt{2k}} 
\right) =
\begin{cases}
0  & \text{if $x < 0$} \\
\frac{b}{2} & \text{if $x = 0$} \\
b e^{- \varepsilon x} & \text{if $x > 0$}
\end{cases} 
\, \, .
\end{equation*}
Since
\begin{equation*}
c = \int_{\mathbb{R}} C_{\sigma} \, dv^1 = a > 0 \, \, , 
\end{equation*}
we conclude from Theorem~\ref{besselrepresentationtheorem} that for $t \in {\mathbb{R}}$
\begin{align} \label{specialbesselrepresentationjumps}
& \left[\overline{\cos \left(t \sqrt{\phantom{ij}} \right)}\,
\bigg|_{\sigma(A_C)}\right]\!(A_C) f = \sum_{k=0}^{\infty} \frac{1}{2^k k!} ( \sqrt{a t^2/ \rho}\,)^{k+1} j_{k-1}
( \sqrt{a t^2 / \rho} \,) \, f_k  \, \, , \nonumber \\
& 
\left[\, \overline{\frac{\sin \left(t \sqrt{\phantom{ij}} \right)}{\sqrt{\phantom{ij}}}} \, \bigg|_{\sigma(A_C)}\right]\!(A_C)
f = t \sum_{k=0}^{\infty} \frac{1}{2^k k!} ( \sqrt{a t^2/ \rho}\,)^{k} j_{k}
( \sqrt{a t^2 / \rho} \,) \, f_k \, \, , 
\end{align}
where 
\begin{align*}
f_0 & := b \, e^{- \varepsilon . {\mathrm{id}}_{\mathbb{R}}} \cdot \chi_{_{[0,\infty)}} \, \, , \\
f_k(x) & := \frac{2 b}{\pi} \, e^{\left(
\frac{\varepsilon \sigma}{2} \sqrt{2k} \right)^2}
\, e^{- \varepsilon x} \, \mathrm{erfc}\!\left(
\frac{\varepsilon \sigma}{2}\sqrt{2k} - \frac{x}{\sigma \sqrt{2k}} 
\right) \\
& \, \, = \frac{b}{\sqrt{2 \pi} \sqrt{k \sigma^2}} \, e^{- \varepsilon x} \int_{-\infty}^{x} e^{- u^2 / (2k \sigma^2)}  \cdot e^{ \varepsilon u} \, du 
\end{align*}
for every $x \in {\mathbb{R}}$ and $k \in {\mathbb{N}}^{*}$, 
and where $A_C$ is as in Lemma~\ref{governingoperator}. 

\end{ex}

In the classical wave equation, as expected, discontinuities propagate
along the characteristics; see Figures~\ref{fig4} and \ref{fig6}
for vanishing initial velocity and displacement, respectively.
On the other hand, in the nonlocal case, the discontinuity remains 
in the same place for all time; see Figures~\ref{fig5} and \ref{fig7}
for vanishing initial velocity and displacement, respectively.
This confirms the results given in \cite{wecknerAbeyaratne2005}.

\section{Conclusion}
\label{sec:conclusion}

Our result that the governing operator is a bounded function of the
classical local operator for scalar-valued functions should be
generalizable to vector-valued case.  Our notable result that the
governing operator $A_C$ of the peridynamic wave equation is a bounded
function of the classical governing operator has far
reaching consequences.  It enables the comparison of peridynamic
solutions to those of classical elasticity.  The remarkable
implication is that it opens the possibly of defining peridynamic-type operators on bounded
domains as functions of the corresponding classical operator.  Since
the classical operator is defined through \emph{local} boundary
conditions, the functions inherit this knowledge. This observation
opens a gateway to incorporate local boundary conditions into nonlocal
theories, which has vital implications for numerical treatment of
nonlocal problems.  This is the subject of our companion paper
\cite{aksoyluBeyerCeliker2014_bounded}.

We expect that the expansions in Theorems \ref{additiontheorem} and
\ref{additiontheorem2} can be used for obtaining the large time
asymptotic of solutions of the nonlocal wave equation.  In the
classical case, as expected, we observe the propagation of waves along
characteristics.  In the nonlocal case, we observe oscillatory
recurrent wave separation.  We think that this phenomenon is worth
investigating.  On the other hand, we observe that discontinuity
remains stationary in the nonlocal case, whereas, it is well-known that
discontinuities propagate along characteristics.  We hold that this
fundamentally difference is one of the most distinguishing feature of
PD.  In conclusion, we believe that we added valuable tools to the of
arsenal of methods to analyze nonlocal problems.

\appendix

\section{Some Proofs from Section \ref{sec:optreatment}}
\subsection{Instability of Solutions}

We give a proof of Theorem~\ref{instability}.

\begin{proof}
Since $\sigma(A)$ is bounded from below, we can define 
\begin{equation*}
\lambda_0 := \inf \{\lambda \in \sigma(A)\} \, \, .
\end{equation*}
Furthermore, since $\sigma(A)$ is closed, $\lambda_0 \in \sigma(A)$ 
and since 
\begin{equation*}
\sigma(A) \cap (-\infty,0) \neq \emptyset \, \, , 
\end{equation*}
we conclude that $\lambda_0 < 0$. Furthermore, let $f \in C({\mathbb{R}},{\mathbb{R}})$ such that $f|_{[0,\infty)}$ is bounded and such that
$f|_{(-\infty,0]}$ is positive and decreasing. In particular, this implies that 
$f|_{\sigma(A)} \in U^s(\sigma(A))$ and also that $f^2|_{(-\infty,0]}$ is positive and decreasing. 
Furthermore, let $0 < \varepsilon < |\lambda_0|$. Then there is 
$\xi \in D(A)$ such that 
\begin{equation*}
\eta := (\chi_{_{[\lambda_0,\lambda_0 + \varepsilon]}}|_{\sigma(A)})(A) \xi
\neq 0_{X}  \, \, .
\end{equation*}
Otherwise, since $D(A)$ is in particular dense in $X$, 
\begin{equation*}
(\chi_{_{[\lambda_0,\lambda_0 + \varepsilon]}}|_{\sigma(A)})(A)
= 0_{L(X,X)} \, \, , 
\end{equation*}
in contradiction to the fact that $\lambda_0 \in \sigma(A)$.
In particular, since  
\begin{equation*}
(\chi_{_{[\lambda_0,\lambda_0 + \varepsilon]}}|_{\sigma(A)})(A) 
\end{equation*}
and $A$ commute, it follows that $\eta \in D(A)$.

Furthermore, 
\begin{align*}
& \|(f|_{\sigma(A)})(A)\eta\|^2  = \braket{(f|_{\sigma(A)})(A)\eta|(f|_{\sigma(A)})(A)\eta} =
\braket{\eta|(f|_{\sigma(A)})^2(A)\eta} \\
& = \braket{\xi|(\chi_{_{[\lambda_0,\lambda_0 + \varepsilon]}}|_{\sigma(A)})(A)(f|_{\sigma(A)})^2(A)\xi} 
=
\int_{\sigma(A)} \chi_{_{[\lambda_0,\lambda_0 + \varepsilon]}} \cdot 
f^2 \, d\psi_{\xi} \\
& \geqslant \int_{\sigma(A)} \chi_{_{[\lambda_0,\lambda_0 + \varepsilon]}} \cdot 
[f(\lambda_0 + \varepsilon)]^2 \, d\psi_{\xi} 
= [f(\lambda_0 + \varepsilon)]^2
\int_{\sigma(A)} \chi_{_{[\lambda_0,\lambda_0 + \varepsilon]}} \, d\psi_{\xi} \\
& = [f(\lambda_0 + \varepsilon)]^2 \cdot \|\eta\|^2 .
\end{align*}
In particular, we conclude that 
\begin{align*}
\left\|\left[\overline{\cos \left(t \sqrt{\phantom{ij}} \right)}\,
\bigg|_{\sigma(A)}\right]\!(A) \eta \right\| \geqslant 
\cosh(t \sqrt{|\lambda_0 + \varepsilon|}\,) \cdot \|\eta\|
\end{align*} 
for all $t \in {\mathbb{R}}$. Since $\varepsilon$ is otherwise 
arbitrary, the latter implies that 
\begin{align*}
\left\|\left[\overline{\cos \left(t \sqrt{\phantom{ij}} \right)}\,
\bigg|_{\sigma(A)}\right]\!(A) \eta \right\| \geqslant 
\cosh(t \sqrt{|\lambda_0|}\,) \cdot \|\eta\| \geqslant
\frac{1}{2} \, e^{t \sqrt{|\lambda_0|}} \cdot \|\eta\| .
\end{align*} 
\end{proof}

\subsection{Solutions of Inhomogeneous Wave Equations}
We give a proof of Theorem~\ref{solutionoftheinhomogenousequation}.

\begin{proof}
In a first step, we note for $\lambda > 0$ that 
\begin{align*} 
& \overline{\frac{\sin \left((t - \tau) \sqrt{\phantom{ij}} \right)}{\sqrt{\phantom{ij}}}}\,(\lambda) = \frac{\sin[(t - \tau) \sqrt{\lambda}\,]}{\sqrt{\lambda}} \\
& =
\frac{\sin(t\sqrt{\lambda} \, )}{\sqrt{\lambda}} \, 
\cos(\tau \sqrt{\lambda} \, ) - \cos(t\sqrt{\lambda} \, )
\, \frac{\sin(\tau \sqrt{\lambda} \, )}{\sqrt{\lambda}} \\
& = \overline{\frac{\sin \left(t \sqrt{\phantom{ij}} \right)}{\sqrt{\phantom{ij}}}}\,(\lambda) \cdot \overline{\cos \left(\tau \sqrt{\phantom{ij}} \right)}\,(\lambda) -
\overline{\cos \left(t \sqrt{\phantom{ij}} \right)}\,(\lambda)
 \cdot \overline{\frac{\sin \left(\tau \sqrt{\phantom{ij}} \right)}{\sqrt{\phantom{ij}}}}\,(\lambda) \, \, .
\end{align*}
Since
\begin{equation*}
\overline{\frac{\sin \left((t - \tau) \sqrt{\phantom{ij}} \right)}{\sqrt{\phantom{ij}}}} \, \, , \, \, \overline{\frac{\sin \left(t \sqrt{\phantom{ij}} \right)}{\sqrt{\phantom{ij}}}} \, \, , \, \,
\overline{\cos \left(\tau \sqrt{\phantom{ij}} \right)} \, \, , \, \,
\overline{\cos \left(t \sqrt{\phantom{ij}} \right)} \, \, , \, \,\overline{\frac{\sin \left(\tau \sqrt{\phantom{ij}} \right)}{\sqrt{\phantom{ij}}}} \, \, ,
\end{equation*}
are entire functions, this implies that 
\begin{align*}
& \overline{\frac{\sin \left((t - \tau) \sqrt{\phantom{ij}} \right)}{\sqrt{\phantom{ij}}}}\,(\lambda) \\
& = \overline{\frac{\sin \left(t \sqrt{\phantom{ij}} \right)}{\sqrt{\phantom{ij}}}}\,(\lambda) \cdot \overline{\cos \left(\tau \sqrt{\phantom{ij}} \right)}\,(\lambda) -
\overline{\cos \left(t \sqrt{\phantom{ij}} \right)}\,(\lambda)
 \cdot \overline{\frac{\sin \left(\tau \sqrt{\phantom{ij}} \right)}{\sqrt{\phantom{ij}}}}\,(\lambda)  \, \, , 
\end{align*}
for every $\lambda \in {\mathbb{C}}$ and hence, by application 
of the spectral 
theorem for densely-defined, self-adjoint linear operators
in Hilbert spaces, that 
\begin{align*}
& \left[\, \overline{\frac{\sin \left((t - \tau) \sqrt{\phantom{ij}} \right)}{\sqrt{\phantom{ij}}}} \, \bigg|_{\sigma(A)}\!\right]\!(A) f(\tau) \\
& = \left\{ 
\left[\, \overline{\frac{\sin \left(t \sqrt{\phantom{ij}} \right)}{\sqrt{\phantom{ij}}}} \, \bigg|_{\sigma(A)}\!\right]\!(A)
\left[\overline{\cos \left(\tau \sqrt{\phantom{ij}} \right)}\, \bigg|_{\sigma(A)}\!\right]\!(A) \right. \\
& \quad \quad - \left.
\left[\overline{\cos \left(t \sqrt{\phantom{ij}} \right)}\,(\lambda)\, \bigg|_{\sigma(A)}\!\right]\!(A)
 \left[\overline{\frac{\sin \left(\tau \sqrt{\phantom{ij}} \right)}{\sqrt{\phantom{ij}}}}\, \bigg|_{\sigma(A)}\!\right]\!(A) \right\} 
 f(\tau) \\
& = b(t) a(\tau) f(\tau) - a(t) b(\tau) f(\tau)
\end{align*}
for all $t,\tau \in {\mathbb{R}}$, 
where $a, b : {\mathbb{R}} \rightarrow L(X,X)$ are defined by 
\begin{align*}
a(t) := \left[\overline{\cos \left(t \sqrt{\phantom{ij}} \right)}\,(\lambda)\, \bigg|_{\sigma(A)}\!\right]\!(A) \, \, , \, \, 
b(t) := \left[\, \overline{\frac{\sin \left(t \sqrt{\phantom{ij}} \right)}{\sqrt{\phantom{ij}}}}\, \bigg|_{\sigma(A)}\!\right]\!(A) \, \, , 
\end{align*}
for every $t \in {\mathbb{R}}$. In the following, for $\xi \in D(A)$, we are going to
use that the maps
\begin{align*}
(\,{\mathbb{R}} \rightarrow X, t \mapsto a(t) \xi\,) \, \, \, \, \text{and} \, \, \, \,
(\,{\mathbb{R}} \rightarrow X, t \mapsto b(t) \xi\,)
\end{align*}
are differentiable with derivatives 
\begin{align*}
(\,{\mathbb{R}} \rightarrow X, t \mapsto - b(t) A \xi\,) \, \, \, \, \text{and} \, \, \, \, 
(\,{\mathbb{R}} \rightarrow X, t \mapsto a(t) \xi\,) \, \, ,
\end{align*}
respectively.
We note that, as a consequence of 
the spectral theorem for densely-defined, self-adjoint linear operators
in Hilbert spaces, that $a$, $b$ are strongly continuous and that 
\begin{equation*}
a(t) D(A) \subset D(A) \, \, , \, \, 
b(t) D(A) \subset D(A) \, \, ,
\end{equation*}
for every $t \in {\mathbb{R}}$. Also for every
$k \in U^s_{\mathbb{C}}(\sigma(A))$, $k(A)D(A) \subset D(A)$ and 
for $\xi \in D(A)$
\begin{equation*}
\|k(A) \xi\|_{A}^2 = \|k(A) \xi\|^2 + \|A k(A) \xi\|^2 =
\|k(A) \xi\|^2 + \|k(A) A \xi\|^2 \, \, \left[ \, \leqslant 
\|k(A)\|_{\text{Op}}^2 \cdot \|\xi\|_{A}^2 \, 
\right]\, \, .
\end{equation*}
Hence $a, b$ induce strongly continuous maps
from ${\mathbb{R}}$ to $X_{A}$, which we indicate with the same symbols, and where $X_{\!A} := (D(A),\|\,\,\|_{A})$. In addition, we note that the inclusion $\iota$ of $X_{A}$ into $X$
is continuous. 
In the next step, we observe for a strongly continuous
$c : {\mathbb{R}} \rightarrow L(X_{\!A},X_{\!A})$ and a continuous 
$g : {\mathbb{R}} \rightarrow X_{\!A}$ that 
\begin{align*}
& \|c(t+h) g(t+h) - c(t) g(t)\|_{A} \\
& = \|c(t+h) g(t+h) - c(t+h) g(t) + c(t+h) g(t) - c(t) g(t)\|_{A} \\
& = \|c(t+h) [g(t+h) - g(t)]_{A} + [c(t+h) - c(t)] g(t)\|_{A} \\
& \leqslant \|c(t+h)\| \cdot \|g(t+h) - g(t)\|_{A} + \|c(t+h)g(t) - c(t) g(t)\|_{A}
\end{align*}
and hence that $({\mathbb{R}} \rightarrow X_{\!A}, t \mapsto c(t) g(t))$
is continuous as well as that 
\begin{align*}
\left( {\mathbb{R}} \rightarrow X_{\!A}, t \mapsto \int^A_{I_t} c(\tau) g(\tau) d\tau \right) \, \, , 
\end{align*}
where $\int^A$ denotes weak integration in $X_{\!A}$,
is differentiable with derivative
\begin{equation*}
({\mathbb{R}} \rightarrow X_{\!A}, t \mapsto c(t) g(t)) \, \, .
\end{equation*}
We conclude for every $t \in {\mathbb{R}}$ that 
\begin{align*}
& b(t) \int_{I_t}^{A} a(\tau) f(\tau) \, d\tau - 
a(t) \int_{I_t}^{A} b(\tau) f(\tau) \, d\tau \\
& =\int_{I_t}^{A} [b(t) a(\tau) f(\tau) - a(t) b(\tau) f(\tau)] \, d\tau \\
& = \int_{I_t}^{A} \left[\, \overline{\frac{\sin \left((t - \tau) \sqrt{\phantom{ij}} \right)}{\sqrt{\phantom{ij}}}} \, \bigg|_{\sigma(A)}\!\right]\!(A) f(\tau) \, d\tau = v(t) \, \, .
\end{align*}
Furthermore, we observe for
$c : {\mathbb{R}} \rightarrow L(X,X)$, $g : {\mathbb{R}} \rightarrow X$such that $\textrm{Ran}(g) \subset D(A)$, 
$t \in {\mathbb{R}}$ and $h \in {\mathbb{R}}^{*}$ that
\begin{align*}
& \frac{1}{h} \, [c(t+h) g(t+h) - c(t) g(t)] \\
& = 
\frac{1}{h} \, [c(t+h) g(t+h) - c(t+h) g(t) + c(t+h) g(t) - c(t) g(t)] \\
& = 
c(t+h) \frac{1}{h}\,[g(t+h) - g(t)] + \frac{1}{h}\,[c(t+h) - c(t)] g(t) \\
& = c(t) \frac{1}{h}\,[g(t+h) - g(t)] + 
\frac{1}{h}\,[c(t+h) g(t) - c(t) g(t)] \\
& \, \quad + [c(t+h) - c(t)] \frac{1}{h}\,[g(t+h) - g(t)]
\end{align*}
and hence that 
\begin{align*}
& \frac{1}{h} \, [a(t+h) g(t+h) - a(t) g(t)] - a(t) g^{\, \prime}(t) +
b(t) A g(t) \\
& = a(t) \left\{ \frac{1}{h}\,[g(t+h) - g(t)] - g^{\, \prime}(t) 
\right\} + 
\frac{1}{h}\,[a(t+h) g(t) - a(t) g(t)] + b(t) A g(t) \\
& \, \quad + [a(t+h) -a(t)] \left\{ \frac{1}{h}\,[g(t+h) - g(t)] 
- g^{\, \prime}(t) \right\} + [a(t+h) -a(t)] g^{\, \prime}(t) \, \, , \\
& \frac{1}{h} \, [b(t+h) g(t+h) - b(t) g(t)] - b(t) g^{\, \prime}(t) -
a(t) g(t) \\
& = b(t) \left\{ \frac{1}{h}\,[g(t+h) - g(t)] - g^{\, \prime}(t) 
\right\} + 
\frac{1}{h}\,[b(t+h) g(t) - b(t) g(t)] - a(t) g(t) \\
& \, \quad + [b(t+h) - b(t)] \left\{ \frac{1}{h}\,[g(t+h) - g(t)] 
- g^{\, \prime}(t) \right\} + [b(t+h) - b(t)] g^{\, \prime}(t) \, \, .
\end{align*}
This implies that 
\begin{equation*}
({\mathbb{R}} \rightarrow X, t \mapsto a(t) g(t)) \, \, , \, \, 
({\mathbb{R}} \rightarrow X, t \mapsto b(t) g(t))
\end{equation*}
are differentiable with derivatives 
\begin{equation*}
({\mathbb{R}} \rightarrow X, t \mapsto a(t) g^{\, \prime}(t) 
- b(t) A g(t)) \, \, , \, \, 
({\mathbb{R}} \rightarrow X, t \mapsto b(t) g^{\, \prime}(t) +  a(t) g(t)) \, \, ,
\end{equation*}
respectively.
Application of the latter to $v$ gives for $t \in {\mathbb{R}}$
 
\begin{align*}
& v^{\, \prime}(t) = b(t) a(t) f(t) 
+ a(t) \int_{I_t}^{A} a(\tau) f(\tau) \, d\tau - a(t) b(t) f(t) 
+ b(t) A \int_{I_t}^{A} b(\tau) f(\tau) \, d\tau \\
& = a(t) \int_{I_t}^{A} a(\tau) f(\tau) \, d\tau + b(t)  \int_{I_t} b(\tau) A f(\tau) \, d\tau \\
& = a(t) \int_{I_t}^{A} a(\tau) f(\tau) \, d\tau + b(t)  \int_{I_t}^{A} b(\tau) A f(\tau) \, d\tau 
\, \, ,
\end{align*}
where $\int$ denotes weak integration in $X$, 
and that  
\begin{align*}
& v^{\, \prime \prime}(t) \\
& = a(t) a(t) f(t) - b(t) A \int_{I_t}^{A} a(\tau) f(\tau) \, d\tau + b(t) b(t) A f(t) + a(t) \int_{I_t} b(\tau) A f(\tau) \, d\tau \\
& = a(t) a(t) f(t) + b(t) b(t) A f(t) - b(t) A \int_{I_t}^{A} a(\tau) f(\tau) \, d\tau + a(t) A \int_{I_t}^{A} b(\tau) f(\tau) \, d\tau \\
& = f(t) - A v(t) \, \, .
\end{align*}

\end{proof}

\subsection{Conservation Laws Induced by Symmetries}
We give a proof of Theorem~\ref{conservationlaws}.

\begin{proof}
Part~(i):
Let $t \in I$ and $h \in {\mathbb{R}}$ such 
that $t + h \in I$. Then 
\begin{align*}
& \frac{j_{u,v}(t+h) - j_{u,v}(t)}{h} \\
& = h^{-1} \left[ 
\braket{u(t+h)|v^{\prime}(t+h)} - \braket{u^{\prime}(t+h)|v(t+h)} - \braket{u(t)|v^{\prime}(t)} + \braket{u^{\prime}(t)|v(t)} \right] \\
& =  h^{-1}\left[ 
\braket{u(t+h) - u(t)|v^{\prime}(t+h)} + 
\braket{u(t)|v^{\prime}(t+h) - v^{\prime}(t)} \right. \\
& \left.  \qquad\, \, \, \, \, \, \, \,
- \braket{u^{\prime}(t+h)|v(t+h) - v(t)} 
- \braket{u^{\prime}(t+h) - u^{\prime}(t)|v(t)} \right]  \, \, . 
\end{align*}
Hence it follows that $j_{u,v}$ is differentiable in $t$ with derivative 
\begin{align*}
j_{u,v}^{\prime}(t) & = \braket{u(t)|(v^{\prime})^{\prime}(t)} -
 \braket{(u^{\prime})^{\prime}(t)|v(t)} = \braket{u(t)|(v^{\prime})^{\prime}(t)} -
 \braket{(u^{\prime})^{\prime}(t)|v(t)} \\
& = - \braket{u(t)|Av(t)} + \braket{A u(t)|v(t)}= 0 \, \, .
\end{align*}
From the latter, we conclude that the derivative of 
$j_{u,v}$ vanishes and hence that  
$j_{u,v}$ is a constant function. 
\newline
Part~(ii):
Since  
$A \circ B \supset B \circ A$, it follows that  
$B(D(A)) \subset D(A)$. Hence $B \circ u$ is a twice continuously differentiable map 
assuming values in $D(A)$ and 
satisfying
\begin{equation*} 
(B \circ u)^{\, \prime \prime}(t) = B u^{\, \prime \prime}(t) =
- B A \, u(t) = - A B \, u(t) = - A (B \circ u)(t)
\end{equation*}
for all $t \in {\mathbb{R}}$. According to Part~(i) this implies that 
$j_{u,B \circ u} : {\mathbb{R}} \rightarrow {\mathbb{C}}$, defined by 
\begin{equation*}
j_{u,B \circ u}(t) := \braket{u(t)|B u^{\prime}(t)} - \braket{u^{\prime}(t)|B u(t)} 
\end{equation*}
for every $t \in {\mathbb{R}}$, is constant.
\newline
Part~(iii): For the proof, let $U_B : {\mathbb{R}} \rightarrow L(X,X)$ be 
the strongly continuous one-parameter group that is generated by $B$. This implies
that 
\begin{align*}
& D(B) = 
\{
\xi \in X: \lim_{t \rightarrow 0, t \neq 0} \frac{1}{t} (U_B(t) - {\textrm{id}}_{X}\!)\xi \, \, 
\textrm{exists}
\} \, \, , \\
& B \xi = \frac{1}{i} \lim_{t \rightarrow 0, t \neq 0} \frac{1}{t} (U_B(t) - {\textrm{id}}_{X}\!)\xi
\end{align*}
for every $\xi \in D(B)$. Since $A$ and $B$ commute, every $f(B)$, where $f \in U_{\mathbb{C}}^s(\sigma(B))$ and $\sigma(B)$ denotes the spectrum of $B$, commutes with 
$A$, i.e., satisfies 
\begin{equation*}
A \circ f(B) \supset f(B) \circ A \, \, .
\end{equation*}
Hence it follows from Part~(ii) that 
\begin{equation*}
j_{u,f_s(B)}(t) := \braket{u(t)|f_s(B) u^{\prime}(t)} - \braket{u^{\prime}(t)|f_s(B) u(t)} 
\end{equation*}
for every $t \in {\mathbb{R}}$, is constant, where 
\begin{align*}
f_s(\lambda) := \frac{1}{i s} \left(e^{i s \lambda}  - 1
\right)
\end{align*}
for every $\lambda \in \sigma(B)$ and $s>0$. Also,  
since $A$ and $B$ commute, every $g(A)$, where $g \in U_{\mathbb{C}}^s(\sigma(A))$ and $\sigma(A)$ denotes the spectrum of $A$, commutes with 
$B$, i.e., satisfies 
\begin{equation*}
B \circ g(A) \supset g(A) \circ B \, \, ,
\end{equation*}
which implies that 
\begin{equation*}
g(A)(D(B)) \subset D(B) 
\end{equation*}
and hence also that 
\begin{equation*}
g(A)(D(A) \cap D(B)) \subset D(A) \cap D(B) \, \, . 
\end{equation*}
Therefore, we conclude from Theorem~\ref{abstractwaveequation},
since $u(0), u^{\prime}(0) \in D(A) \cap D(B)$, that 
$\textrm{Ran}(u), \textrm{Ran}(u^{\prime})  \subset D(A) \cap  D(B)$. As a consequence,
for every $t \in {\mathbb{R}}$, 
\begin{equation*}
\lim_{s \rightarrow 0} j_{u,f_s(B)}(t) = \braket{u(t)| B u^{\prime}(t)} - \braket{u^{\prime}(t)|B u(t)} \, \, . 
\end{equation*}
Finally, since $j_{u,f_s(B)}$ is a constant function for $s > 0$, we conclude that 
\begin{equation*}
j_{u,B}(t) := \braket{u(t)| B u^{\prime}(t)} - \braket{u^{\prime}(t)|B u(t)} 
\end{equation*}
for every $t \in {\mathbb{R}}$, is a constant function. 
\end{proof}

\pagebreak

\bibliographystyle{siam}
\bibliography{../../../data/digest/gemFEM2011/paper3_stokes/bib/burak}

\end{document}